\documentclass[11pt, a4paper]{article}

\usepackage[english]{babel}
\usepackage[utf8]{inputenc}
\usepackage{enumerate,amssymb,amsmath,amsthm}
\usepackage[all,cmtip,arrow,matrix,curve]{xy}
\usepackage{abstract}
\usepackage{geometry}
\usepackage{tikz}
\usepackage{authblk}

\newtheorem{thm}{Theorem}[section]

\newtheorem{cor}[thm]{Corollary}
\newtheorem{lem}[thm]{Lemma}

\newtheorem{defn}[thm]{Definition}

\providecommand{\del}{\partial}
\providecommand{\rd}{\mathrm{d}}

\newcommand{\Z}{\mathbb{Z}}
\newcommand{\N}{\mathbb{N}}
\newcommand{\R}{\mathbb{R}}

\numberwithin{equation}{section}
\usepackage{hyperref}

\begin{document}

\title{Special solutions to a non-linear coarsening model with local interactions}

\author{Constantin Eichenberg}

\affil{Institute for Applied Mathematics, University of Bonn \\
Endenicher Allee 60, 53115 Bonn, Germany \\
\texttt{eichenberg@iam.uni-bonn.de}}

\maketitle

\begin{abstract}
\noindent We consider a class of mass transfer models on a one-dimensional lattice with nearest-neighbour interactions. The evolution is given by the backward parabolic equation $\del_t x = - \frac{\beta}{|\beta|} \Delta x^\beta$, with $\beta$ in the fast diffusion regime $(-\infty,0) \cup (0,1]$. Sites with mass zero are deleted from the system, which leads to a coarsening of the mass distribution. The rate of coarsening suggested by scaling is $t^\frac{1}{1-\beta}$ if $\beta \neq 1$ and exponential if $\beta = 1$. We prove that such solutions actually exist by an analysis of the time-reversed evolution. In particular we establish positivity estimates and long-time equilibrium properties for discrete parabolic equations with initial data in $\ell_+^\infty(\Z)$.
\end{abstract}

{\small \tableofcontents}

\section{Introduction}

Discrete mass transfer models with local interactions have been studied by several authors in different contexts. They have several applications in physics such as the growth and coarsening of sand ripples in \cite{HellenOK2002} or the clustering in granular gases \cite{VanDerMeervdWL2001}, while also serving as approximations or toy models for more complex coarsening scenarios such as the evolution of droplets in dewetting films \cite{GlasnerW2003,GlasnerW2005} and grain growth \cite{HenselerNO2004}. If the mass transport between sites is symmetric, the evolution of such systems in one dimension is governed by an infinite system of ODEs,

\begin{align} \label{mass_transfer_equation}
\frac{\rd}{\rd t} x(t,k) = F(x(t,k-1)) - 2F(x(t,k)) + F(x(t,k+1)), 
\end{align}

\noindent where the right hand side represents the net mass flux at a site $k$ which receives and transfers mass from its neighbours at rates controlled by the flux function $F$. This system can also be interpreted as the spatially discrete non-linear PDE $\del_t x = \Delta F(x)$. The monotonicity properties of $F$ are crucial for the qualitative behaviour of solutions and depend on the application, as an increasing flux function will lead to mass diffusion and a decreasing flux function will lead to aggregation and coarsening. A combination of both is also possible, for example in models that were investigated in \cite{EsedogluS2008,EsedogluG2009}. \newline

\noindent In this paper we are interested in the coarsening model proposed in \cite{HelmersNV2016}, with flux function

\begin{align} 
F(x) = F_\beta(x) = - \frac{\beta}{|\beta|}x^\beta,
\end{align}

\noindent where $ -\infty < \beta < 0$ or $0 < \beta \leq 1$. This largely resembles the sand ripple scenario \cite{HellenOK2002}, although we will refer to the lattice points as \textit{particles} from now on. Distinctive features of the model are the infinite number of particles and the \textit{vanishing rule:} \textit{Particles that reach mass zero are deleted from the system and the remaining particles are relabeled accordingly}. This way small particles vanish from the system while transferring their mass to the rest of the system, which leads to a growth of the average particle size and an overall coarsening of the system. \newline

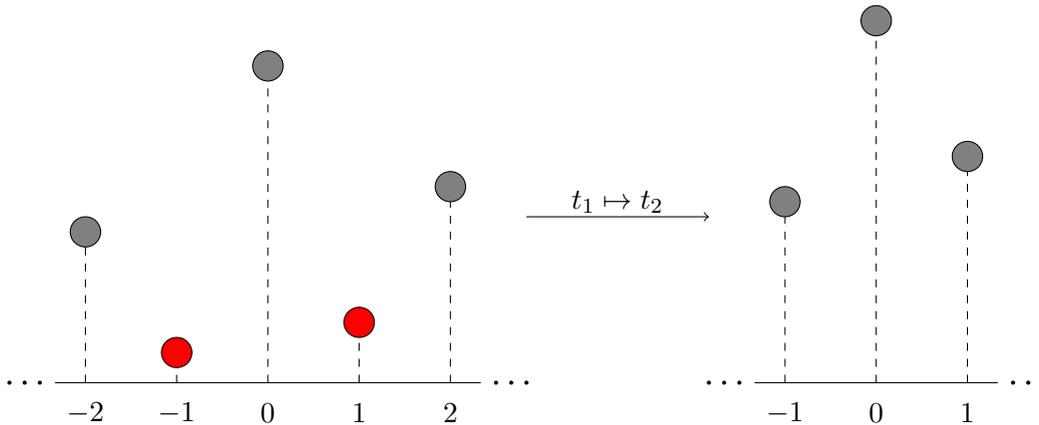
\begin{figure}
\begin{center}
\begin{tikzpicture}[scale = 0.4]
\draw (-20.5,-10) -- (-6.5,-10);
\draw[dashed] (-19.5,-10) -- (-19.5,-5.5);
\draw[dashed] (-16.5,-10) -- (-16.5,-9.5);
\draw[dashed] (-13.5,-10) -- (-13.5,0);
\draw[dashed] (-10.5,-10) -- (-10.5,-8.5);
\draw[dashed] (-7.5,-10) -- (-7.5,-4);

\draw[fill = gray]  (-19.5,-5) circle (0.5);
\draw[fill = red]  (-16.5,-9) circle (0.5);
\draw[fill = gray]  (-13.5,0.5) circle (0.5);
\draw[fill = red]  (-10.5,-8) circle (0.5);
\draw[fill = gray]  (-7.5,-3.5) circle (0.5);

\draw[fill = black]  (-6,-10) circle (0.05);
\draw[fill = black]  (-5.5,-10) circle (0.05);
\draw[fill = black]  (-5,-10) circle (0.05);
\draw[fill = black]  (-21,-10) circle (0.05);
\draw[fill = black]  (-21.5,-10) circle (0.05);
\draw[fill = black]  (-22,-10) circle (0.05);

\node at (-19.5,-11) {$ -2$};
\node at (-16.5,-11) {$ -1$};
\node at (-13.5,-11) {$ 0$};
\node at (-10.5,-11) {$ 1$};
\node at (-7.5,-11) {$ 2$};


\node at (-2,-4) {$t_1 \mapsto t_2$};
\draw[->] (-5,-4.5) -- (1,-4.5);

\draw[fill = black]  (1,-10) circle (0.05);
\draw[fill = black]  (1.5,-10) circle (0.05);
\draw[fill = black]  (2,-10) circle (0.05);

\draw (2.5,-10) -- (10.5,-10);
\draw[dashed] (3.5,-10) -- (3.5,-4.5);
\draw[dashed] (6.5,-10) -- (6.5,1.5);
\draw[dashed] (9.5,-10) -- (9.5,-3);
\draw[fill = gray]  (3.5,-4) circle (0.5);
\draw[fill = gray]  (6.5,2) circle (0.5);
\draw[fill = gray]  (9.5,-2.5) circle (0.5);

\draw[fill = black]  (11,-10) circle (0.05);
\draw[fill = black]  (11.5,-10) circle (0.05);
\draw[fill = black]  (12,-10) circle (0.05);

\node at (3.5,-11) {$-1$};
\node at (6.5,-11) {$ 0$};
\node at (9.5,-11) {$ 1$};
%
\end{tikzpicture}
\end{center}
\caption{Small particles vanish at $t_2$ and the average mass increases.}
\end{figure}

\noindent With this particular choice of the flux function (except when $\beta = 1$), the equation has an invariant scaling: If $x = x(t,k)$ is a solution, then 

\begin{align}
x_\lambda(t,k) = \lambda^\frac{1}{\beta - 1}x(\lambda t,k)
\end{align}

\noindent is another solution. Thus, if $\langle x \rangle$ denotes a suitable lenght-scale, we expect that

\begin{align} \label{scaling_law}
\langle x \rangle \sim t^\frac{1}{1-\beta}.
\end{align}

\noindent In the case $\beta = 1$ the mean-field analysis in \cite{HellenOK2002} indicates that $\langle x \rangle \sim \exp(\lambda t)$, where $\lambda$ is not universal but depends on the initial distribution. \newline

\noindent The problem in the mathematical analysis of such models is to rigorously establish such coarsening rates. The method of Kohn and Otto \cite{KohnO2002} has proved very useful in several situations to obtain upper bounds. Here, $\langle x \rangle$ is usually some negative Sobolev norm. In our setting however, their technique is not obviously applicable. Although the system is formally an $H^{-1}$ gradient flow, the corresponding energy is infinite due to the presence of infinitely many particles. Additionally, the vanishing and relabeling of particles is problematic in this context. Still, the simple structure of our model enables us to apply more elementary arguments to derive upper bounds:  For positive $\beta$, the right hand side of equation (\ref{mass_transfer_equation}) can be estimated to obtain

\begin{align}
\dot{x} \leq 2x^\beta,
\end{align}

\noindent which can be integrated to yield the desired bound in the $\ell^\infty$-norm. For negative $\beta$ the equation gives $\dot{x} \geq -2x^\beta$, which can be used to derive a weak upper bound, see Proposition 2.4 in \cite{HelmersNV2016}. Furthermore, the numerical simulations and heuristics in \cite{HelmersNV2016} demonstrate that single particles can grow linearly (thus faster than the scaling law) in time, showing that an $\ell^\infty$-bound cannot be expected in this case. \newline

\noindent On the other hand, not much is known about the validity of lower bounds. As will be demonstrated below, there are many non-constant initial configurations which become stationary after a finite time due to the vanishing of particles. An easy example for this is a 2-periodic configuration of large and small particles. During the evolution, the large particles grow and the small particles shrink until disappearing at the same time, at which all large particles will be left with the same size and the evolution stops.

\begin{figure}
\begin{center}
\begin{tikzpicture}[scale = 0.4]
\draw (-20.5,-10) -- (2.5,-10);
\draw[dashed] (-19.5,-10) -- (-19.5,-7);
\draw[dashed] (-16.5,-10) -- (-16.5,-4);
\draw[dashed] (-13.5,-10) -- (-13.5,-7);
\draw[dashed] (-10.5,-10) -- (-10.5,-4);
\draw[dashed] (-7.5,-10) -- (-7.5,-7);
\draw[dashed] (-4.5,-10) -- (-4.5,-4);
\draw[dashed] (-1.5,-10) -- (-1.5,-7);
\draw[dashed] (1.5,-10) -- (1.5,-4);
\draw[fill = red]  (-19.5,-6.5) circle (0.5);
\draw[fill = gray]  (-16.5,-3.5) circle (0.5);
\draw[fill = red]  (-13.5,-6.5) circle (0.5);
\draw[fill = gray]  (-10.5,-3.5) circle (0.5);
\draw[fill = red]  (-7.5,-6.5) circle (0.5);
\draw[fill = gray]  (-4.5,-3.5) circle (0.5);
\draw[fill = red]  (-1.5,-6.5) circle (0.5);
\draw[fill = gray]  (1.5,-3.5) circle (0.5);
\draw[fill = black]  (3,-10) circle (0.05);
\draw[fill = black]  (3.5,-10) circle (0.05);
\draw[fill = black]  (4,-10) circle (0.05);
\draw[fill = black]  (-21,-10) circle (0.05);
\draw[fill = black]  (-21.5,-10) circle (0.05);
\draw[fill = black]  (-22,-10) circle (0.05);
\end{tikzpicture}
\end{center}
\caption{After the smaller particles have vanished, the configuration is constant.}
\end{figure}
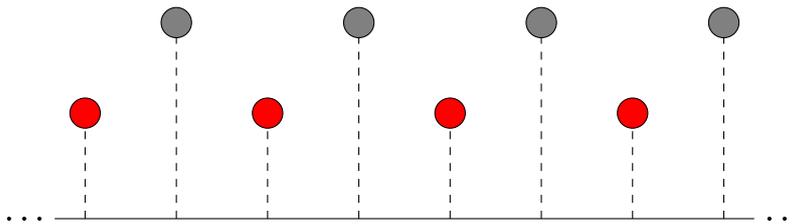

\noindent The problem of classifying all initial data for which some form of a lower coarsening bound holds is completely open. The main result of this work is the existence of initial data and corresponding solutions with scale-characteristic coarsening rates, where $\langle x \rangle$ is a suitable average of the configuration, see Theorem \ref{main_result}. Our general  ansatz is to reverse time, which transforms the equation into a non-linear discrete parabolic equation which behaves much better and can be analysed by means of Harnack-type positivity estimates (see \cite{BonforteV2006}) and parabolic regularity theory (see \cite{Nash1958} for the continuum theory and \cite{GiacominOS2001} for the discrete analogue). It should be mentioned that the solutions that we construct coarsen in a very organised manner, whereas numerical simulations and heuristics that were done in \cite{HelmersNV2016} indicate that the generic coarsening behaviour is more disorganised. Nevertheless we believe that this result is useful as a first step to a better understanding of the coarsening dynamics of this model. In particular, our methods imply the instability of constant configurations in a rather strong sense, see Corollary \ref{instability_result}. \newline

\noindent The rest of the paper is organised as follows: In Section 2 we give a precise description of the model and our main result. A general technique for the construction of solutions is presented in Section 3, while the main result is proved in Section 4. We postpone the proofs of some necessary technical results such as existence of solutions for the time-reversed setting, Harnack inequalities and discrete Nash-Aronson estimates to the appendix.

\section{Statement of results}

We collect some definitions and introduce appropriate notation to give a rigorous description of our setting. We largely follow \cite{HelmersNV2016}, where this model was first introduced in a mathematical context. Then we state our main result and give a short outline of the proof.

\subsection{Setup and notation}

We consider a discrete infinite number of particles with non-negative mass on a one-dimensional lattice. That means each configuration is an element of the space 

\begin{align}
\ell_+^\infty(\Z) = \left \{ x = x(k) \in \ell^\infty(\Z): \ x(k) \geq 0   \right \}.
\end{align}

\noindent As described above, particles with zero mass will be deleted from the system during the evolution. However, relabeling the particle indices whenever a particle vanishes can be problematic. On the one hand, relabeling can be ambiguous, for example the vanishing times might not be in order or could have accumulation points. On the other hand the solution will be discontinuous in time. Thus it is more convenient to leave the configuration unchanged and update the interaction term on the right-hand side of equation (\ref{mass_transfer_equation}) instead. For this purpose we define the the \textit{nearest living neighbour indices} 

\begin{align}
\sigma_+(x,k) &= \inf \{l > k: x(l) > 0 \}, \\
\sigma_-(x,k) &= \sup \{l < k: x(l) > 0 \},
\end{align}

\noindent where we just write $\sigma_{\pm}(k)$ if there is no danger of confusion. Also we define the ordinary discrete Laplacian $\Delta$ and the \textit{living particles Laplacian} $\Delta_\sigma$ as

\begin{align}
\Delta x &= x(k-1) - 2x(k) + x(k), \\
\Delta_\sigma x (k) &= \left( x(\sigma_-(k)) - 2x(k) + x(\sigma_+(k)) \right) \cdot \chi_{\{x(k) \neq 0 \}}.
\end{align} 

\noindent Then the evolution of the system is governed by the following equation:

\begin{align} \label{coarsening_equation}
\begin{cases}
\del_t x = \Delta_\sigma F_\beta(x) \quad \mathrm{in} \ (0,\infty) \times \Z, \\
x(0,\cdot) = x_0, 
\end{cases}
\end{align}

\noindent with $x_0 \in \ell_+^\infty$ and

\begin{align} 
F_\beta(x) = - \frac{\beta}{|\beta|}x^\beta,
\end{align}

\noindent with $F_\beta(0) := 0$ for $ \beta < 0$. The only drawback is that the right-hand side of (\ref{coarsening_equation}) is no longer continuous, hence we have to use a concept of mild solutions, as in \cite{HelmersNV2016}: \newline

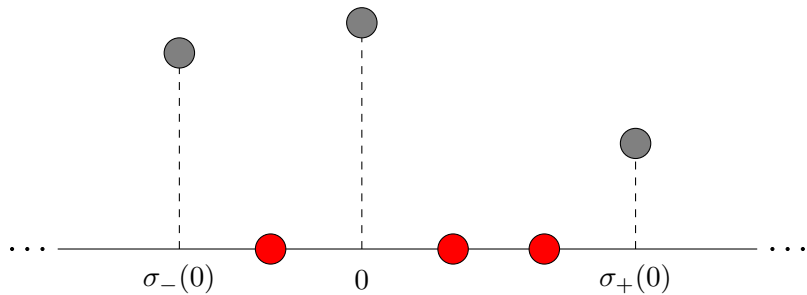
\begin{figure}
\begin{center}
\begin{tikzpicture}[scale = 0.4]
\draw (-20.5,-10) -- (2.5,-10);
\draw[dashed] (-16.5,-10) -- (-16.5,-4);
\draw[dashed] (-10.5,-10) -- (-10.5,-3);
\draw[dashed] (-1.5,-10) -- (-1.5,-7);
\draw[fill = gray]  (-16.5,-3.5) circle (0.5);
\draw[fill = red]  (-13.5,-10) circle (0.5);
\draw[fill = gray]  (-10.5,-2.5) circle (0.5);
\draw[fill = red]  (-7.5,-10) circle (0.5);
\draw[fill = red]  (-4.5,-10) circle (0.5);
\draw[fill = gray]  (-1.5,-6.5) circle (0.5);
\draw[fill = black]  (3,-10) circle (0.05);
\draw[fill = black]  (3.5,-10) circle (0.05);
\draw[fill = black]  (4,-10) circle (0.05);
\draw[fill = black]  (-21,-10) circle (0.05);
\draw[fill = black]  (-21.5,-10) circle (0.05);
\draw[fill = black]  (-22,-10) circle (0.05);

\node at (-10.5,-11) {$0$};
\node at (-1.5,-11) {$\sigma_+(0)$};
\node at (-16.5,-11) {$\sigma_-(0)$};
\end{tikzpicture}
\end{center}
\caption{Vanished particles remain in the physical domain, only neighbour relations $\sigma_+, \sigma_-$ change.}
\end{figure}

\begin{defn}
Let $ 0 < T \leq \infty$. We say that $x: [0,T) \to \ell_+^\infty(\Z)$ is a solution to problem (\ref{coarsening_equation}) if the following conditions are satisfied:

\begin{enumerate}
\item $t \mapsto x(t,k)$ is continuous on $[0,T)$ and $x(0,k) = x_0(k)$ for every $k \in \Z$.

\item $t \mapsto F_\beta(x)\cdot \chi_{\{x(k) \neq 0 \}}$ is locally integrable on $[0,T)$ for every $k \in \Z$.

\item For every $0 \leq t_1 < t_2 < T$ and every $k \in \Z$ we have 

\begin{align}
x(t_1,k) - x(t_2,k) = \int_{t_1}^{t_2} \Delta_\sigma F_\beta(x)(s,k) \ \rd s.
\end{align}
\end{enumerate}

\end{defn}

\noindent The second condition is automatically satisfied if $\beta$ is positive. For the existence of solutions we refer to \cite{HelmersNV2016}, where the case $\beta < 0$ is discussed. We expect a similar result to hold for positive $\beta$ but since we are only concerned with special solutions anyway we will give no proof here. More important for our result is the well-posedness of the time-reversed evolution

\begin{align} 
\del_t u &= \Delta G_\beta(u), 
\end{align}

\noindent with $G_\beta(u) = -F_\beta(u)$, which is the discrete analogue of a fast diffusion equation. This is addressed in the appendix, see Theorem \ref{appendix_existence_thm}. \newline

\noindent It is easy to check that the evolution (\ref{coarsening_equation}) conserves the average mass

\begin{align}
\langle x \rangle = \lim_{N \to \infty} \frac{1}{2N+1}\sum_{k=-N}^N x(k).
\end{align}

\noindent This is not really meaningful, since vanished and living particles are treated the same. To adequately measure the coarsening process, one has to average only over the living particles. Consequently we define

\begin{align}
L_{\sigma_+,N} &= \bigcup_{k = 1}^N \{ \sigma_+^k(0) \}, \\
L_{\sigma_-,N} &= \bigcup_{k = 1}^N \{ \sigma_-^k(0) \}, \\
\end{align}

\noindent as sets of the first $N$ positive, respectively negative living particle particle indices and set

\begin{align}
L_{\sigma,N} &= \begin{cases}
&L_{\sigma_+,N} \cup L_{\sigma_-,N},\  \mathrm{if} \ x(0) = 0, \\
&L_{\sigma_+,N} \cup L_{\sigma_-,N} \cup {0}, \  \mathrm{if} \ x(0) > 0.
\end{cases}
\end{align}

\noindent Then we can define the \textit{living particle means} 

\begin{align}
\langle x \rangle_{\sigma,N} &= \frac{1}{|L_{\sigma,N}|}\sum_{k \in L_{\sigma,N}} x(k), \\
\langle x \rangle_{\sigma}^+ &= \limsup_{N \to \infty} \ \langle x \rangle_{\sigma,N}, \\
\langle x \rangle_{\sigma}^- &= \liminf_{N \to \infty} \ \langle x \rangle_{\sigma,N}.
\end{align}

\noindent Since mass is transferred from small to large particles and the small particles eventually vanish, we expect the living particle means to grow in time.\newline

\subsection{Main result}

In the main result of the paper we show that there exist solutions where the average particle size grows with the characteristic rate that is indicated by scaling:

\begin{thm} \label{main_result}

Let $\beta \in (-\infty,0) \cup (0,1]$ and $F_\beta$ be defined as above. Then the following statements hold:

\begin{enumerate}
\item For every $\beta \in (-\infty,0) \cup (0,1)$ there exists $x_0 \in \ell_+^\infty(\Z)$ and a solution to equation (\ref{coarsening_equation}) with initial data $x_0$ that satisfies

\begin{align}
\langle x \rangle_{\sigma}^- &\gtrsim t^{\frac{1}{1-\beta}}, \\
||x||_\infty &\lesssim t^{\frac{1}{1-\beta}}.
\end{align}

\item For $\beta = 1$ there exists $ 0 < \lambda \leq 2$, $x_0 \in \ell_+^\infty(\Z)$ and a solution to equation (\ref{coarsening_equation}) with initial data $x_0$ that satisfies

\begin{align}
\langle x \rangle_{\sigma}^- &\gtrsim \exp(\lambda t), \\
||x||_\infty &\lesssim \exp(\lambda t).
\end{align}

\end{enumerate}

\noindent Here, $\gtrsim$ and $\lesssim$ mean that the corresponding inequalities hold up to a multiplicative constant that depends only on $\beta$.
\end{thm}

\noindent In the following we give a short outline of the proof. The key observation is that the time-reversed system corresponding to equation (\ref{coarsening_equation}) is a non-linear parabolic equation where particles are inserted instead of vanishing, which is much easier to handle. Thus the idea is to make a more or less explicit construction in the time-reversed setting and then reverse time again to obtain a sequence of approximate solutions $x^{(n)}$ which solve (\ref{coarsening_equation}) and eventually converge to a solution with the desired properties. Each solution $x^{(n)}$ will be constructed in $n$ steps, starting in the future time $T_n$ (with $T_n \to \infty$), where the particle sizes are of order $\theta^{n}$ for some $\theta > 1$. We then insert particles to lower the average particle size to order $\theta^{n-1}$ and run the time-reversed evolution, equilibrating the system until all particle sizes are of order $\theta^{n-1}$. The procedure is then iterated, going from sizes of order $\theta^{n+1-j}$ to $\theta^{n-j}$, until after $n$ steps all particles sizes are of order one. A suitable compactness argument for $n \to \infty$ then yields a solution $x$ on $[0,\infty)$ to equation (\ref{coarsening_equation}). \newline

\noindent In order to achieve the desired coarsening rate the time-span to equilibrate in the $j$-th step has to be of order $\theta^{(1-\beta)(n+1-j)}$, which is a-priori not clear. Due to scaling however, every step is equivalent to the problem of inserting particles into a configuration $u_0$ of order one (denoted by $u_0 \mapsto \Psi_* u_0$) such that after evolving the system under the backward equation for a uniform timespan $T$ the particles are of order $\theta^{-1}$. More precisely, we will prove the following result, which is the heart of the argument:

\begin{lem}[Key Lemma]\label{key_lemma}
Let $\beta \in (-\infty,0) \cup (0,1]$ and $G_\beta = -F_\beta$. Then for every $\varepsilon > 0$ there exists $T=T(\beta,\varepsilon) > 0$, such that the following holds: For every $u_0 \in \ell_+^\infty(\Z)$ with $\frac{1}{2} \leq u_0 \leq 1$ there exists a creation operator $\Psi_*$ and a solution $u$ of the equation

\begin{align}\label{backward_equation}
\begin{cases}
\del_t u = \Delta G_\beta(u) \quad \mathrm{in} \ (0,\infty) \times \Z,  \\
u(0,\cdot) = \Psi_*u_0,
\end{cases}
\end{align}

\noindent that satisfies

\begin{align}
\left| u(T,.) - \frac{1}{2} \right| \leq \varepsilon.
\end{align}

\end{lem}

\begin{figure}
\begin{center}
\begin{tikzpicture}[scale = 0.4]
\draw (-20.5,-10) -- (-12.5,-10);
\draw[dashed] (-19.5,-10) -- (-19.5,-3.5);
\draw[dashed] (-17.5,-10) -- (-17.5,-5);
\draw[dashed] (-15.5,-10) -- (-15.5,-0.5);
\draw[dashed] (-13.5,-10) -- (-13.5,-5.5);

\draw[orange,dashed] (-20.5,0) -- (9.5,0);
\draw[orange,dashed] (-20.5,-5) -- (9.5,-5);

\node at (-22.5,0) {$\theta^{n+1-j}$};
\node at (-22.5,-5) {$\frac{1}{2}\theta^{n+1-j}$};

\draw[fill = gray]  (-19.5,-3) circle (0.5);
\draw[fill = gray]  (-17.5,-4.5) circle (0.5);
\draw[fill = gray]  (-15.5,0) circle (0.5);
\draw[fill = gray]  (-13.5,-5) circle (0.5);

\draw[fill = black]  (-12,-10) circle (0.05);
\draw[fill = black]  (-11.5,-10) circle (0.05);
\draw[fill = black]  (-11,-10) circle (0.05);
\draw[fill = black]  (-21,-10) circle (0.05);
\draw[fill = black]  (-21.5,-10) circle (0.05);
\draw[fill = black]  (-22,-10) circle (0.05);

\node at (-9,-3.5) {$\Psi_*$};
\draw[->] (-11,-4.5) -- (-7,-4.5);

\draw (-6.5,-10) -- (9.5,-10);
\draw[dashed] (-5.5,-10) -- (-5.5,-3.5);
\draw[dashed] (-1.5,-10) -- (-1.5,-5);
\draw[dashed] (2.5,-10) -- (2.5,-0.5);
\draw[dashed] (8.5,-10) -- (8.5,-5.5);

\draw[fill = gray]  (-5.5,-3) circle (0.5);
\draw[fill = red]  (-3.5,-10) circle (0.5);
\draw[fill = gray]  (-1.5,-4.5) circle (0.5);
\draw[fill = red]  (0.5,-10) circle (0.5);
\draw[fill = gray]  (2.5,0) circle (0.5);
\draw[fill = red]  (4.5,-10) circle (0.5);
\draw[fill = red]  (6.5,-10) circle (0.5);
\draw[fill = gray]  (8.5,-5) circle (0.5);

\draw[fill = black]  (-8,-10) circle (0.05);
\draw[fill = black]  (-7.5,-10) circle (0.05);
\draw[fill = black]  (-7,-10) circle (0.05);
\draw[fill = black]  (10,-10) circle (0.05);
\draw[fill = black]  (10.5,-10) circle (0.05);
\draw[fill = black]  (11,-10) circle (0.05);

\node (v1) at (7.5,-12.5) {};
\node (v2) at (-0.5,-17) {};
\draw[->]  (v1) edge (v2);

\draw (-20.5,-21.5) -- (-4.5,-21.5);

\draw[orange,dashed] (-20.5,-15.5) -- (-4.5,-15.5);
\draw [orange,dashed](-20.5,-18.5) -- (-4.5,-18.5);

\draw[dashed] (-19.5,-21.5) -- (-19.5,-16.5);
\draw[dashed] (-17.5,-21.5) -- (-17.5,-19);
\draw[dashed] (-15.5,-21.5) -- (-15.5,-17);
\draw[dashed] (-13.5,-21.5) -- (-13.5,-18.5);
\draw[dashed] (-11.5,-21.5) -- (-11.5,-16);
\draw[dashed] (-9.5,-21.5) -- (-9.5,-17.5);
\draw[dashed] (-7.5,-21.5) -- (-7.5,-18);
\draw[dashed] (-5.5,-21.5) -- (-5.5,-17);

\draw[fill = gray]  (-19.5,-16) circle (0.5);
\draw[fill = gray]  (-17.5,-18.5) circle (0.5);
\draw[fill = gray]  (-15.5,-16.5) circle (0.5);
\draw[fill = gray]  (-13.5,-18) circle (0.5);
\draw[fill = gray]  (-11.5,-15.5) circle (0.5);
\draw[fill = gray]  (-9.5,-17) circle (0.5);
\draw[fill = gray]  (-7.5,-17.5) circle (0.5);
\draw[fill = gray]  (-5.5,-16.5) circle (0.5);

\draw[fill = black]  (-22,-21.5) circle (0.05);
\draw[fill = black]  (-21.5,-21.5) circle (0.05);
\draw[fill = black]  (-21,-21.5) circle (0.05);
\draw[fill = black]  (-3.5,-21.5) circle (0.05);
\draw[fill = black]  (-3,-21.5) circle (0.05);
\draw[fill = black]  (-2.5,-21.5) circle (0.05);

\node at (7.5,-15.5) {$\Delta t \sim \theta^{(1-\beta)(n+1-j)}$};
\node at (-22.5,-15.5) {$\theta^{n-j}$};
\node at (-22.5,-18.5) {$\frac{1}{2}\theta^{n-j}$};
\end{tikzpicture}
\end{center}
\caption{The $j$-th step in the back-in-time construction.}
\end{figure}
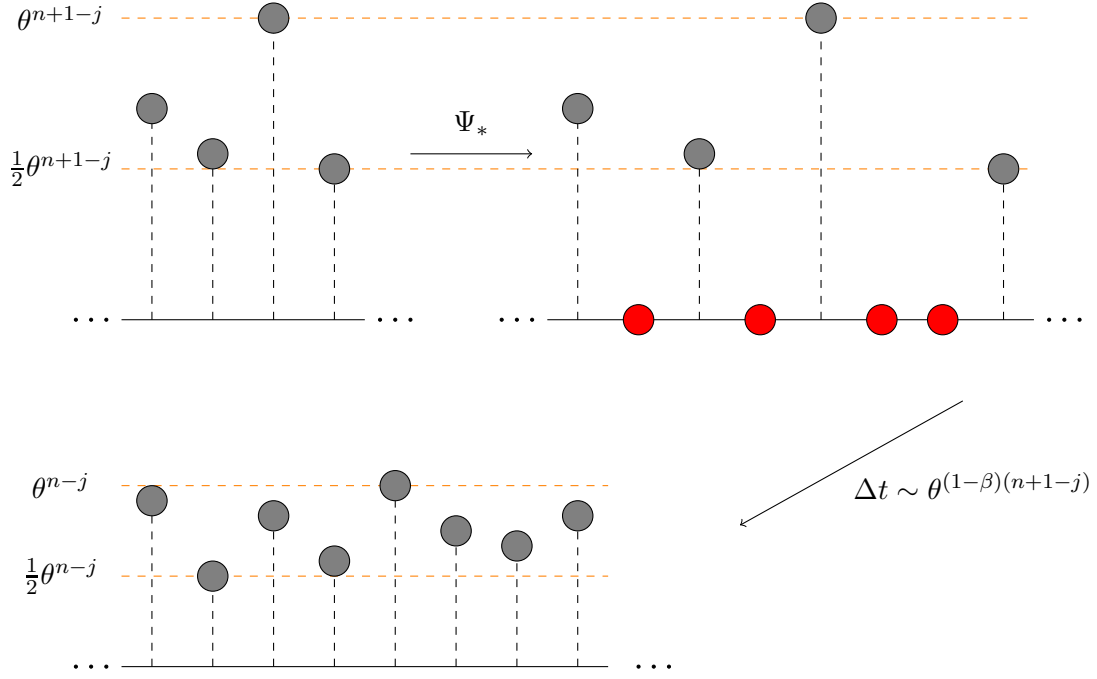

\noindent The precise meaning of $\Psi_* u_0$ will be explained in the next section. In particular, if we set $\theta^{-1} = 1/2 + \varepsilon$ with $\varepsilon \leq 1/6$ then $u$ will satisfy the desired estimate

\begin{align}
\frac{1}{2}\theta^{-1} \leq u(T,.) \leq \theta^{-1}.
\end{align}

\noindent The main idea to prove the lemma is to insert particles such that $1/2 - \varepsilon \leq \Psi_*u_0 \leq 1/2 + \varepsilon$ holds in an averaged sense. Since the backward equation is a diffusion, it is expected that the system equilibrates and average-wise estimates eventually induce point-wise estimates after a certain timespan, see Lemma \ref{long_time_diffusivity_lemma}. Note that due to the freedom of choice in the parameter $\varepsilon$, the back-in-time construction can generate initial data that are arbitrarily flat, demonstrating the instability of constant data:

\begin{cor}\label{instability_result}
Let $c > 0$. Then for every $\varepsilon > 0$ there exist initial data $x_0$ as in Theorem \ref{main_result} such that 

\begin{align}
||x_0 - c||_\infty \leq \varepsilon.
\end{align}

\end{cor}

\noindent Before proving these results we introduce the formalism $\Psi_*$ for the insertion of particles and explain the general construction of solutions to the coarsening equation. \newline

\section{Construction of solutions}

The general idea to construct (local-in-time) solutions to equation (\ref{coarsening_equation}) is to choose some terminal data $x(T,\cdot)$ and go backward in time from there. The crucial observation is that the vanishing of particles corresponds to the \textit{creation} of particles if time is reversed. Additionally, since the living particles do not carry any information of the vanished particles in the forward-in-time equation, new particles can be created at arbitrary times $\tau_j$ and positions $\{\Psi_*^{(j)} \}$ (for notation see below) in the backward equation. This gives the necessary freedom to construct solutions with desirable properties. Summarizing the above considerations, each data triple $(x(T,\cdot),\{ \tau_j \}, \{\Psi_*^{(j)} \})$ gives rise to a solution of equation (\ref{coarsening_equation}) on the interval $[0,T]$. In the following section we carry this out in detail. 

\subsection{Insertion of particles}

First we fix the notation for the insertion of particles. Basically, we need a precise way to insert zeroes into a given sequence of numbers. The most practical way to do this is via push-forward of a suitable increasing map $\Psi: \Z \to \Z$. This map can be defined by the corresponding sequence of "jumps". We make the following definition:

\begin{defn}
Let $d:\Z \to \N_0$ be a given sequence of jumps. Then define the corresponding deformation as

\begin{align}
&\Psi: \Z \to \Z \\
&\Psi(k) = k + \sum_{m=0}^k d(m).
\end{align} 

\noindent Now for every $x \in \ell_+^\infty(\Z)$ we define the push-forward sequence $\Psi_* x$ as

\begin{align}
\Psi_* x(\Psi(k)) = x(k),
\end{align}

\noindent for all $k \in \Z$ and $\Psi_* x(l) = 0$ if $ l \notin \mathrm{Im}(\Psi)$.

\end{defn} 

\noindent With this definition, if $x$ represents a particle configuration, then $\Psi_* x$ represents the same configuration with new mass-zero particles created. To be more precise, the condition $d(k) = l$ exactly means that we are inserting $l$ new particles between the $k$-th and the $(k\pm 1)$-th particle, (depending on the sign of $k$). We will refer to the mapping $\Psi_*$ as \textit{particle creation operator} and, to keep notation as compact as possible, not explicitly refer to the deformation $\Psi$ or the specific jump sequence $d$ any more, but rather just state where particles are inserted. This is potentially ambiguous, for instance, "creating a particle between each two living particles" can be achieved by different $d$, potentially translating the original living particles. However, in the following sections these ambiguities do not affect the arguments, hence we will ignore them. 

\subsection{Back-in-time construction}

Next we describe how to obtain a solution from a given terminal configuration $x_{\mathrm{ter}}$, an increasing sequence of vanishing/creation times $\{ \tau_j \}_{j=1,..,n}$ and corresponding creation operators $\{\Psi_*^{(j)} \}_{j=1,..,n} $. We define the solution piecewise by iteratively using the backward evolution (\ref{backward_equation}) on $[\tau_{j-1},\tau_j]$ after inserting particles at $t = \tau_{j-1}$ and continuing the procedure. To be precise, we define $u^{(j)}$ on the interval $[\tau_{j-1},\tau_j]$ to be a solution of the following problem:

\begin{align} \label{iteration_step}
&\begin{cases}
\del_t u^{(j)} = \Delta G_\beta(u^{(j)}) \ \mathrm{in} \ (\tau_{j-1}, \tau_j] \times \Z, \\
u^{(j)}(\tau_{j-1}) = \Psi_*^{(j)}\left[u^{(j-1)}(\tau_{j-1})\right],
\end{cases}
\end{align}

\noindent for $j=1,..,n$, with $\tau_0 := 0$, $u^{(0)}(\tau_0) := x_{\mathrm{ter}}$ and $G_\beta = -F_\beta$. We should note that by a solution we mean a classical solution, i.e $u^{(j)} \in C^0([\tau_{j-1},\infty),\ell_+^\infty(\Z))$, for every $k \in \Z$ we have $u^{(j)}(.,k) \in C^1((\tau_{j-1},\infty))$ and the equation holds pointwise. Well-posedness of this problem is a-priori not clear, especially for the case $\beta < 0$. For the moment we just assume that the equation is solvable and focus on carrying out the construction of solutions to the coarsening equation. In Theorem \ref{appendix_existence_thm} we give a sufficient condition on the initial data for existence of solutions that is easily verified for the data considered in the next section. \newline

\noindent Reversing the time direction we obtain piecewise solutions of our original equation. However, one has to compose $u^{(j)}$ with the creation operators once more, since vanished particles remain in the "physical" domain in the original evolution (\ref{coarsening_equation}). To be more precise, we set

\begin{align} 
x^{(j)}(t) &= \left( \prod_{l = 1}^{j-1} \Psi_*^{(n+1-l)} \right) \left[u^{(n+1-j)}(\tau_{n} - t) \right], 
\end{align}

\noindent which lets us glue the solutions together in a continuous way:

\begin{align}
x(t) &= x^{(j)}(t), \ \mathrm{if} \ t \in [\tau_n - \tau_{n+1-j},\tau_n - \tau_{n-j} ), 
\end{align}

\noindent for $j=1,..,n$. Using $u^{(j)}(\tau_{j-1}) = \Psi_*^{(j)}\left[u^{(j-1)}(\tau_{j-1})\right]$ it is easy to check that $x$ defined this way is continuous in time. The next lemma shows that $x$ is indeed a solution to our original equation:

\begin{lem}
Let $\Psi_*$ be a creation operator as above. Then we have

\begin{itemize}
\item[1.]$\sigma_{\pm}(\Psi_*x,\Psi(k)) = \Psi(\sigma_\pm(x,k))$ for every $x \in \ell_+^\infty$ and $k \in \Z$. \\
\item[2.]$[\Delta_\sigma,\Psi_*]x = \left(\Delta_\sigma \Psi_* - \Psi_* \Delta_\sigma \right)x = 0$ for every $x \in \ell_+^\infty$. \\
\item[3.]$\langle \Psi_*x \rangle_{\sigma,N} = \langle x \rangle_{\sigma,N} $ for every $N > 0$ and $x \in \ell_+^\infty$. 
\end{itemize}

\end{lem}

\begin{proof}
\textit{1.} It suffices to prove the claim for $\sigma_+$, the other case is completely analogous. Because $\Psi$ is strictly increasing, we have $\Psi(\sigma_+(x,k)) > \Psi(k)$. We also have

\begin{align}
\Psi_*x(\Psi(\sigma_+(x,k))) = x(\sigma_+(x,k)) > 0,
\end{align}

\noindent which shows $\sigma_{+}(\Psi_*x,\Psi(k)) \leq \Psi(\sigma_+(x,k))$. For the other inequality, we note that $\Psi_*x(l) > 0$ implies that $l = \Psi(m)$ for some $m \in \Z$. In this case we have

\begin{align}
0 < \Psi_*x(l) = x(m)
\end{align}

\noindent which implies $m \geq \sigma_+(x,k)$, and because $\Psi$ is increasing we conclude

\begin{align}
l = \Psi(m) \geq \Psi(\sigma_+(x,k)),
\end{align}

\noindent which proves the first assertion.\newline

\noindent \textit{2.} Let $l = \Psi(k)$. We apply the identity in \textit{1.} to get

\begin{align}
\Delta_\sigma \Psi_* x (l) &= \left( \Psi_* x(\sigma_-(\Psi_* x,l)) - 2\Psi_* x(l) + \Psi_* x(\sigma_+(\Psi_* x,l)) \right) \cdot \chi_{\{\Psi_* x(l) \neq 0 \}} \\
&= \left( x(\sigma_-(k)) - 2x(k) + x(\sigma_+(k)) \right) \cdot \chi_{\{x(k) \neq 0 \}} \\
 &= \Delta_\sigma x(k) = \Psi_* \Delta_\sigma x (l).
\end{align}

\noindent On the other hand, if $l \notin \mathrm{Im}(\Psi)$, the identity is trivial. \newline

\noindent \textit{3.} Obvious from the definition.

\end{proof}

\noindent With the second statement of the above lemma, it is not difficult to verify that the sequence $x$ we have constructed above solves equation (\ref{coarsening_equation}):

\begin{cor}
Let $x_\mathrm{ter}$, $\{ \tau_j \}$ and $\{\Psi_*^{(j)} \}$ be given and $x$ be constructed as above. If $t \mapsto F_\beta(x)\cdot \chi_{\{x(k) \neq 0 \}}$ is locally integrable for every $k \in \Z$, then $x$ is a (mild) solution to equation (\ref{coarsening_equation}) on $[0,\tau_n)$.
\end{cor}

\begin{proof}
Since $x$ is continuous and piecewise smooth by construction, it suffices to show that $\del_t x = \Delta_\sigma F_\beta(x)$ holds pointwise on all intervals $[\tau_n - \tau_{n+1-j},\tau_n - \tau_{n-j} )$. Indeed, we calculate

\begin{align}
\del_t x^{(j)}(t) &= \left( \prod_{l = 1}^{j-1} \Psi_*^{(n+1-l)} \right) \left[\del_t u^{(n+1-j)}(\tau_{n} - t) \right] \\
&= \left( \prod_{l = 1}^{j-1} \Psi_*^{(n+1-l)} \right) \left[\Delta_\sigma F_\beta(u^{(n+1-j)})(\tau_{n} - t) \right] \\
&= \Delta_\sigma F_\beta \left(\left( \prod_{l = 1}^{j-1} \Psi_*^{(n+1-l)} \right) \left[u^{(n+1-j)}(\tau_{n} - t) \right] \right) = \Delta_\sigma F_\beta \left(x^{(j)} \right).
\end{align}

\noindent Here we used that $\Delta_\sigma$ commutes with creation operators by the previous lemma, as well as composition with the function $F_\beta$.
\end{proof}

\noindent \textbf{Remark}. The above construction scheme implies the existence of many initial data and corresponding solutions to the coarsening equation which become stationary after a finite time. Indeed, $x$ as above has this property if we pick $x_\mathrm{ter}$ to be a constant sequence. Because there is much freedom in the choice of particle creations and vanishing times this means that finding conditions on initial data such that lower coarsening bounds hold is a difficult task and remains an open problem. In the construction for the proof of Theorem \ref{main_result} we will in fact choose $x_\mathrm{ter}(k) = \theta^n$ so that each approximate solution becomes stationary. Because $\theta^n \to \infty$ and $\tau_n \to \infty$ the limit solution however will grow indefinitely. The details will be explained in the next section.

\section{Proof of the main result}

We divide the full proof of Theorem \ref{main_result} into four main steps. In the first step we show how to insert particles to modify the local average in a uniform way. The second step is to prove a long-time diffusive property of the backward equation which, together with the first step, will enable us to prove Lemma \ref{key_lemma}. In the third step the construction of the approximate sequence $x^{(n)}$ is thoroughly carried out. Finally we use a compactness argument to pass to the limit and obtain a solution with the desired properties, finishing the proof.

\subsection{Step 1: Average modification by particle insertion}

\begin{defn}[Local Averages]
Let $x \in \ell^\infty(\Z)$. Then define the associated sequence of local averages as

\begin{align}
\Lambda(x,k,N) = \frac{1}{2N+1}\sum_{l=-N}^N x(k-l).
\end{align}

\end{defn}

\noindent In the next lemma we show how to modify the local averages of a given sequence by inserting particles in a suitable way:

\begin{lem}[Particle insertion]\label{particle_insertion_lemma}
Let $u_0 \in \ell_+^\infty(\Z)$ with

\begin{align}
\frac{1}{2} \leq u_0 \leq 1.
\end{align}

\noindent Then for every $\varepsilon > 0$ there exists a creation operator $\Psi_*$ and $N_0 \in \N$ such that

\begin{align}
\left | \Lambda(\Psi_* u_0,.,N) - \frac{1}{2} \right | \leq \varepsilon,
\end{align}

\noindent for $N \geq N_0$. Furthermore, if $d$ is the jump sequence associated to $\Psi_*$, then $||d||_\infty$ is finite and depends only on $\varepsilon$.
\end{lem}

\begin{proof}

Let $(\lambda_i)$ be an equidistant partition of the interval $[1/2,1]$ with $|\lambda_i - \lambda_{i+1}| \leq \varepsilon$. We give an explicit scheme for the particle insertion as follows: We divide $\Z$ into disjoint blocks of particles with length $K$, where $K$ is determined later:

\begin{align}
\Z = \bigcup_{j \in \Z} B_j,
\end{align}

\noindent with $B_j = \{ jK,...,(j+1)K - 1 \}$. Let $\Lambda_j$ denote the average mass in $B_j$ with respect to $u_0$. We define the deformation $\Psi$ by inserting $L_i$ (determined later) particles to the right of $(j+1)K-1$ whenever he have

\begin{align}
\lambda_i \leq \Lambda_j \leq \lambda_{i+1}.
\end{align}

\noindent This gives rise to a new partition of $\Z$ into blocks $\tilde{B}_j$ with varying lenghts $K + L_i$, where $\tilde{B}_j$ contains all elements of $\Psi (B_j)$ and the next $L_i$ numbers that are not elements of $\mathrm{Im}(\Psi)$. We call a block with $L_i$ inserted particles a \textit{block of the $i$-th kind}. Then the average mass $\tilde{\Lambda}_{j}$ of such a block with respect to $\Psi_*u_0$ is by construction 

\begin{align}
\tilde{\Lambda}_j = \frac{1}{K+L_i}\left(\sum_{k \in B_j} u_0(k) \right) = \frac{K}{K+L_i}\Lambda_j,
\end{align}

\noindent which gives
 
\begin{align}
\lambda_{i} \frac{K}{K+L_i} \leq \tilde{\Lambda}_{j} \leq \lambda_{i+1} \frac{K}{K+L_i} := \lambda_{i+1}\theta_i.
\end{align}

\noindent Because $1/2 \leq \lambda_i \leq 1$ we can, if $K$ is large enough, choose $L_i \leq K$ such that 

\begin{align}
\left |\lambda_i \theta_i - \frac{1}{2} \right | \leq \mathcal{O}(\varepsilon),
\end{align} 

\noindent and because $\lambda_i$ and $\lambda_{i+1}$ are close we also have

\begin{align}
\left |\lambda_{i+1} \theta_i - \frac{1}{2} \right | \leq \mathcal{O}(\varepsilon).
\end{align}

\noindent This implies that the average mass of every block $\tilde{B}_j$ can be estimated as

\begin{align}
\left |\tilde{\Lambda}_{j} - \frac{1}{2} \right | \leq \mathcal{O}(\varepsilon).
\end{align}

\noindent Next we calculate $\Lambda(\Psi_* u_0,k,N)$ for $N \gg K$ and arbitrary $k \in \Z$. Denote by $n_i$ the number of blocks of the $i$-th kind in the domain of summation, that is $\{ k-N,....,k+N \}$. This implies that  

\begin{align}
|\{ k-N,....,k+N \}| = 2N + 1 = \sum_i (K+L_i)n_i + \mathcal{O}(K).
\end{align}

\noindent Then we divide the summation in $\Lambda(\Psi_* u_0,k,N)$ into summation over the respective blocks and the rest of the particles in $\{ k-N,....,k+N \}$, whose number, and thus total mass $R$, is of order $K$. Thus we have

\begin{align}
\Lambda(\Psi_* u_0,k,N) &= \frac{1}{2N+1} \left(\sum_{\mathrm{sum \ over \ blocks}} + R \right) \\
&= (1+\mathcal{O}(K/N))\frac{\sum_{\mathrm{sum \ over \ blocks}}}{\sum_i (K+L_i)n_i}  + \mathcal{O}(K/N).
\end{align}

\noindent By the estimates on the average masses of the blocks we have 

\begin{align}
\left( \frac{1}{2} - \varepsilon \right)\sum_i (K+L_i)n_i \leq \sum_{\mathrm{sum \ over \ blocks}} \leq \left( \frac{1}{2} + \varepsilon \right) \sum_i (K+L_i)n_i,
\end{align}

\noindent which implies the desired estimate if $K/N \leq \mathcal{O}(\varepsilon)$. Because $L_i \leq K$ by construction the jump sequence satisfies $d \leq K$ and $K$ depends only on $\varepsilon$.
\end{proof}

\subsection{Step 2: Estimate for the backward equation}

The basic idea to analyse equation (\ref{backward_equation}) is to view it as a discrete parabolic equation in divergence form with time-dependent coefficients. More precisely, with the finite difference operators

\begin{align}
\del^+ u(k) &= u(k+1) - u(k), \\
\del^- u(k) &= u(k) - u(k-1),
\end{align}

\noindent we calculate

\begin{align}
\del_t u = \Delta G_\beta(u) = \del^- \del^+ (G_\beta(u)) = \del^-(a \del^+ u),
\end{align}

\noindent where

\begin{align}
a(t,k) = a_\beta(t,k) = \frac{G_\beta(u(t,k+1)) - G_\beta(u(t,k))}{u(t,k+1) - u(t,k)}.
\end{align}

\noindent The coefficient $a$ is strictly positive and bounded from below if $u$ is bounded from above but becomes singular at $u = 0$, except for $\beta = 1$, where $a = 1$. Because of this we want to work with solutions that are bounded from above and below: 

\begin{lem}(Positivity estimate)\label{backward_equation_short_time_estimate}
Let $\beta \in (-\infty,0) \cup (0,1)$ and $u_0 \in \ell_+^\infty(\Z)$ with $\frac{1}{2} \leq u_0 \leq 1$. Let $\Psi_*$ be a creation operator with associated jump sequence $d$ that satisfies $d(k) \leq L$. Then there exists a (classical) solution to equation (\ref{backward_equation}) with initial data $\Psi_* u_0$. Furthermore, we have $u \leq 1$ and

\begin{align}
u(t,\cdot) \geq c \left( 1 \wedge t^{\frac{1}{1-\beta}} \right), \label{positivity_estimate}
\end{align}

\noindent where $c$ depends only on $\beta$ and $L$.
\end{lem}

\begin{proof}
Because the jump sequence satisfies $d \leq L$, the distance between particles that have mass at least $1/2$ is at most $L+1$. Then the result follows directly from Theorem \ref{appendix_existence_thm}, since the above considerations imply $\Psi_* u_0 \in \mathcal{P}_{L+1,\frac{1}{2}}$.
\end{proof}

\noindent The lemma implies that there exists a solution $u$ such that equation (\ref{backward_equation}) is immediately strictly parabolic. Before we turn to the analysis of linear parabolic equations we establish uniform H\"older continuity. This is important for the stability of local averages for small times and later for the compactness of the approximating sequence.

\begin{lem}(Uniform H\"older continuity)\label{uniform_Hoelder_continuity_lemma}
Let $u_0,\beta, \Psi_*$ and $u$ be as above. Then the following statements hold:

\begin{enumerate}
\item For $\beta < 0$ and $T > 0$ there exists $C = C(\beta,L,T)$ such that 

\begin{align}
|u(t_2,k) - u(t_1,k)| \leq C|t_2 - t_1|^\frac{1}{1-\beta},
\end{align}

for all $0 < t \leq T$ and $k \in \Z$.

\item For $\beta \in (0,1]$ we have 

\begin{align}
|u(t_2,k) - u(t_1,k)| \leq 4|t_2 - t_1|,
\end{align}

for all $t > 0$ and $k \in \Z$.

\end{enumerate}
 
\end{lem}

\begin{proof}
First we note that due to equation (\ref{backward_equation}) we have for $t_2 > t_1$:

\begin{align*}
|u(t_2,k) - u(t_1,k)| \leq \int_{t_1}^{t_2} |\Delta G_\beta(u)(s,k)| \ \rd s.
\end{align*}

\noindent The estimate $u \leq 1$ implies $|\Delta G_\beta(u)| \leq 4$ if $\beta \in (0,1]$. For negative $\beta$ we use the lower bound (\ref{positivity_estimate}) to get $|\Delta G_\beta(u)(s,k)| \leq C(\beta,L,T) s^{\frac{\beta}{1-\beta}}$ on each compact interval $[0,T]$. Then the desired inequality follows by using the estimates on $\Delta G_\beta(u)$ in the above identity and the elementary inequality $a^\frac{1}{1-\beta} - b^\frac{1}{1-\beta} \leq (a-b)^\frac{1}{1-\beta}$ for $a \geq b$ and $\beta < 0$.

\end{proof}

\noindent The next step is the long-time diffusivity result for linear equations, making use of discrete parabolic H\"older regularity (see appendix for details). 

\begin{lem}(Long-time estimate)\label{long_time_diffusivity_lemma}
Let $a:[0,\infty) \to \ell_+^\infty(\Z)$ with $0 < \lambda_1 \leq a \leq \lambda_2$ and $a(.,k) \in C^0([0,\infty))$ for every $k \in \Z$. Let $u_0 \in \ell_+^\infty(\Z)$ and assume that there exist positive constants $c_1,c_2$ such that $c_1 \leq \Lambda(u_0,.,N) \leq c_2$ for $N \geq N_0$. Let $u$ be the solution of

\begin{align}
\begin{cases}
\del_t u = \del^-(a \del^+ u) = \mathcal{L}(t)u \\ \label{discrete_parabolic_equation}
u(0,.) = u_0.
\end{cases}
\end{align}

\noindent Then for every $\varepsilon > 0$ there exists $T = T(\varepsilon) > 0$, depending only on $N_0$ and the bounds on $a$, such that 

\begin{align}
c_1 - \varepsilon \leq u(t,k) \leq c_2 + \varepsilon
\end{align}

\noindent for all $t \geq T$, $ k \in \Z$.

\end{lem}

\begin{proof}
Since spatial translation does not change the type of equation and the bounds on $a$, it suffices to estimate $u(t,0)$. Let $\phi = \phi(t,k,s,l)$ denote the full fundamental solution to equation (\ref{discrete_parabolic_equation}). Then we have

\begin{align}
u(t,0) = \sum_l \phi(t,0,0,l)u_0(l) =  \int_\R U(t,\xi)u_0(\lfloor t^\frac{1}{2} \xi \rfloor) \ \rd \xi,
\end{align}

\noindent with

\begin{align}
U(t,\xi) = t^\frac{1}{2}\phi(t,0,0,\lfloor t^\frac{1}{2} \xi \rfloor).
\end{align}

\noindent Let $\varepsilon > 0$. By Corollary \ref{stepfct_approximation_fund_sol} there exist $T,\delta > 0$ depending only on $\varepsilon$ and the bounds on $a$ such that $U(t,.)$ can be approximated in $L^1$ by step-functions (which can be chosen to be positive since $\phi$ is positive) of step-width $\delta$ up to an error of $\varepsilon$ for $t \geq T$. Let $\chi = \sum_k a_k \chi_{I_k}$ be such a step-function, then we calculate

\begin{align}
\int_\R \chi(\xi)u_0(\lfloor t^\frac{1}{2} \xi \rfloor) \ \rd \xi &= \sum_k a_k \int_{I_k}u_0(\lfloor t^\frac{1}{2} \xi \rfloor) \ \rd \xi \\
&= \sum_k |I_k| a_k \frac{1}{t^\frac{1}{2} |I_k|} \int_{t^\frac{1}{2} I_k}u_0(\lfloor \xi \rfloor) \ \rd \xi.
\end{align}

\noindent Since $c_1 \leq \Lambda(u_0,.,N) \leq c_2$ for $N \geq N_0$, we have

\begin{align}
c_1 - \varepsilon \leq \frac{1}{2R}\int_{[-R,R]}u_0(\lfloor \xi-\eta \rfloor) \ \rd \eta \leq c_2 + \varepsilon,
\end{align}

\noindent for large enough $R$. Hence we have 

\begin{align}
(c_1-\varepsilon)||\chi||_{L^1} \leq \int_\R \chi(\xi)u_0(\lfloor t^\frac{1}{2} \xi \rfloor) \ \rd \xi \leq (c_2+\varepsilon)||\chi||_{L^1},
\end{align}

\noindent for large enough $t \geq T$, since $|I_k| \geq \delta$. For such $t$ and $\chi$ approximating $U$ we calculate

\begin{align}
\int_\R U(t,\xi)u_0(\lfloor t^\frac{1}{2} \xi \rfloor) \ \rd \xi &\geq (c_1-\varepsilon)||\chi||_{L^1} - ||u_0||_\infty||U(t,.) - \chi||_{L^1} \\
&= c_1 + \mathcal{O}(\varepsilon),
\end{align} 

\noindent where we used $||\chi||_{L^1} = ||U(t,.)||_{L^1} + \mathcal{O}(\varepsilon) = 1  + \mathcal{O}(\varepsilon)$ in the last step. The other bound is analogous. 
\end{proof}

\noindent Combining these two results we can prove the Key Lemma \ref{key_lemma}:\newline

\begin{proof}[Proof of Lemma \ref{key_lemma}]
Let $\varepsilon > 0$. By Lemma \ref{particle_insertion_lemma} there exists $\Psi_*$ (with $||d||_\infty$ depending on $\varepsilon$) and $N_0 = N_0(\varepsilon)$ such that 

\begin{align}
\frac{1}{2} - \varepsilon \leq \Lambda(\Psi_* u_0,\cdot,N) \leq \frac{1}{2} + \varepsilon,
\end{align}

\noindent for $N \geq N_0$. Because of uniform H\"older continuity (see Lemma \ref{uniform_Hoelder_continuity_lemma}) there exists $t_1 = t_1(\beta,\varepsilon) > 0$ such that

\begin{align}
\frac{1}{2} - 2\varepsilon \leq \Lambda(u(t_1,\cdot),\cdot,N) \leq \frac{1}{2} + 2\varepsilon.
\end{align}

\noindent On the other hand, Lemma \ref{backward_equation_short_time_estimate} implies that 

\begin{align}
u(t_1,\cdot) \geq \delta = \delta(\beta,\varepsilon) > 0,
\end{align}

\noindent for all $t \geq t_1$. In particular, equation (\ref{backward_equation}) is strictly parabolic on $(t_1,\infty)$ (for $\beta = 1$ this step is obsolete). Then according to Lemma \ref{long_time_diffusivity_lemma} there exists $t_2 = t_2(\beta,\varepsilon)$ such that 

\begin{align}
\frac{1}{2} - 3\varepsilon \leq u(T,\cdot) \leq \frac{1}{2} + 3\varepsilon,
\end{align}

\noindent with $T = t_2 + t_1$. 

\end{proof}

\subsection{Step 3: Approximating sequence}

We will now construct the approximating sequence $x^{(n)}$, using the technique established in Section 3. Thus for every $n$ we have to specify the terminal data, creation times $\tau_{n,j}$ and creation operators $\Psi_*^{(n,j)}$. For the rest of the section we fix $0 < \varepsilon \leq 1/6$ and define 

\begin{align}
\theta^{-1} = \frac{1}{2} + \varepsilon.
\end{align}

\noindent Let $T$ be as in Lemma \ref{key_lemma} and set

\begin{align}
x^{(n)}_\mathrm{ter} = \theta^n,
\end{align}

\noindent a constant sequence. Now $\tau_{n,j}$ and $\Psi_*^{(n,j)}$ are constructed iteratively. We choose $\Psi_*^{(n,1)}$ according to Lemma \ref{key_lemma} applied to the constant sequence $1$, then by the statement of the lemma, the fact that $1/2 - \varepsilon \geq \theta^{-1}/2$ and the scaling properties of the equation we have

\begin{align}
\frac{1}{2}\theta^{n-1} \leq u^{(n,1)}(T\theta^{(1-\beta)n},\cdot) \leq \theta^{n-1},
\end{align}

\noindent where $u^{(n,1)}$ is a solution to the backward equation (\ref{backward_equation}) with initial data $\Psi_*^{(n,1)}x^{(n)}_\mathrm{ter}$ according to Lemma \ref{key_lemma}. Consequently we set 

\begin{align}
\tau_{n,1} = T\theta^{(1-\beta)n}.
\end{align}

\noindent Iterating the procedure, for given $\tau_{n,j-1}$ and $u^{(n,j-1)}$ with

\begin{align}
\frac{1}{2}\theta^{n-j+1} \leq u^{(n,j-1)}(\tau_{n,j-1},.) \leq \theta^{n-j+1},
\end{align}

\noindent we apply Lemma \ref{key_lemma} to the rescaled sequence

\begin{align}
\theta^{-n+j-1}u^{(n,j-1)}(\tau_{n,j-1},.),
\end{align}

\noindent which yields a creation operator $\Psi_* =: \Psi_*^{(n,j)}$ and a solution $u^{(n,j)}$ to 

\begin{align}
&\begin{cases}
\del_t u^{(n,j)} = \Delta G_\beta(u^{(n,j)}) \ \mathrm{in} \ (\tau_{j-1}, \tau_j] \times \Z, \\
u^{(n,j)}(\tau_{n,j-1}) = \Psi_*^{(n,j)}\left[u^{(n,j-1)}(\tau_{n,j-1})\right]
\end{cases},
\end{align}

\noindent that, by scaling, satisfies  

\begin{align}\label{iteration_estimate}
\frac{1}{2}\theta^{n-j} \leq u^{(n,j)}(\tau_{n,j},.) \leq \theta^{n-j},
\end{align}

\noindent with

\begin{align}
\tau_{n,j} &= \tau_{n,j-1} + T\theta^{(1-\beta)(n+1-j)}.
\end{align}

\noindent As described in the previous section, this procedure yields a solution $x^{(n)}$ on the interval $[0,\tau_{n,n}]$ to the coarsening equation. The local integrability condition for negative $\beta$ is satisfied due to (\ref{positivity_estimate}). Since $x^{(n)}(\tau_{n,n})$ is constant up to vanished particles, the solution can be extended to $[0,\infty)$. Let

\begin{align}
T_n &= \tau_{n,n}, \notag \\
t_{j} &= T_n - \tau_{n,n-j} = T \sum_{k=n+1-j}^{n}\theta^{(1-\beta)(n+1-k)} = T\sum_{m=1}^{j}\theta^{(1-\beta)m}. \label{vanishing_time_formula}
\end{align}

\noindent The numbers $t_j$ are exactly the times where particles can vanish and are the same for all $n$. In particular, the vanishing times of the limit will be contained in the set $\{t_j \}$. We summarize the properties of of $x^{(n)}$ that follow directly by construction:

\begin{align}
&\del_t x^{(n)} = \Delta_\sigma F_\beta(x^{(n)}) \ \mathrm{in} \ (0,\infty) \times \Z, \label{approximate_property_1} \\ 
&x^{(n)}(t,k) = \theta^n \ \mathrm{for \ all} \ t \geq T_n \ \mathrm{and} \ x^{(n)}(t,k) > 0, \label{approximate_property_2} \\ 
&x^{(n)}(t,k) \leq \theta^{j} \ \mathrm{for \ all} \ t_{j-1} \leq t \leq t_{j} \ \mathrm{and} \ k \in \Z, \label{approximate_property_3} \\ 
&x^{(n)}(t_j,k) \geq \frac{1}{2}\theta^{j} \ \mathrm{for \ all} \ 1 \leq j \leq n \ \mathrm{and} \ k \in \Z \ \mathrm{with} \ x^{(n)}(t_j,k) > 0. \label{approximate_property_4} 
\end{align}

\subsection{Step 4: Passage to the limit}

\noindent Before using an appropriate compactness argument on the approximating sequence it is also necessary to control the decay of particles near their vanishing times uniformly since the particle interaction is discontinuous at $x(k) = 0$. y\newline

\begin{lem}\label{Vanishing_time_estimate}
Let $x^{(n)}$ be defined as above and $\beta \neq 1$. For every $j > 0$ there exists $C = C(j,\beta)$ and $\varepsilon = \varepsilon(j,\beta) > 0$, such that  for all particles $k \in \Z$ that vanish at $t = t_j$ we have

\begin{align}
x^{(n)}(t,k) \geq C(t_j-t)^{\frac{1}{1-\beta}}, 
\end{align}

\noindent for $t \in [t_j - \varepsilon, t_j]$. For $\beta = 1$ the statement holds in the same way except we have the lower bound

\begin{align}
x^{(n)}(t,k) \geq C \exp(-2(t_j-t))I_L(2(t_j-t)),
\end{align}

\noindent where $I_L$ is the $L$-th modified Bessel function of the first kind and $L$ depends only on $\beta$.

\end{lem}

\begin{proof}

The construction of $x^{(n)}$ and Lemma \ref{backward_equation_short_time_estimate} directly imply the statement for $\beta \neq 1$. The inequality for $\beta = 1$ follows from Theorem \ref{appendix_existence_thm} in the same way as Lemma \ref{backward_equation_short_time_estimate}.
\end{proof}

\noindent With the previous preparation we can prove the main result:

\begin{proof}[Proof of Theorem \ref{main_result}]
Using the explicit construction of the sequence $x^{(n)}$ and Lemma \ref{uniform_Hoelder_continuity_lemma} it is easy to see that $x^{(n)}(\cdot,k)$ is uniformly H\"older continuous on $[0,T]$ for every $k \in \Z$ and $T > 0$ defined above. By the Arzela-Ascoli Theorem we can, after an exhaustion $T_N \nearrow +\infty$ and a diagonal argument, extract a subsequence (not renamed) and a limit $x = x(t,k)$ such that

\begin{align}
x^{(n)}(t,k) \to x(t,k) \ \mathrm{as} \ n \to \infty \ \mathrm{for \ all} \ t \in [0,\infty) \ \mathrm{and} \ k \in \Z,
\end{align}

\noindent and $x(.,k) \in C([0,\infty))$ for all $k \in \Z$. Let

\begin{align*}
\eta_k^{(n)} = \inf \{t > 0: \ x^{(n)}(t,k) = 0 \}
\end{align*}

\noindent denote the vanishing time of the $k$-th particle. By another diagonal argument we can further restrict ourselves to a subsequence such that the particle vanishing times $\{\eta_k^{(n)} \}$ converge:

\begin{align}
\eta_k^{(n)} \to \eta_k \in \{t_j \}_{j = 0,..,\infty} \cup +\infty \ \mathrm{as} \ n \to \infty \ \mathrm{for \ all} \ k \in \Z.
\end{align}

\noindent In fact if $\eta_k^{(n)}$ is bounded we even have $\eta_k^{(n)} = \eta_k$ for $n$ large enough because the set $\{t_j \}$ is discrete. Otherwise we can assume $\eta_k^{(n)} \to +\infty$ increasingly. If $\eta_k = t_j $ we check that this is indeed the vanishing time of $x(.,k)$:

\begin{align}
0 = \lim_{n \to \infty} x^{(n)}(t,k) = x(t,k),
\end{align}

\noindent for every $t > \eta_k $. Furthermore, by Lemma \ref{Vanishing_time_estimate} we have for $\beta \neq 1$

\begin{align}
x^{(n)}(t,k) \geq C(t_j - t)^{\frac{1}{1-\beta}} > 0 \ \mathrm{for \ all } \ t \in [t_j - \varepsilon,t_j),
\end{align}

\noindent for large enough $n$, hence also $x(t,k) > 0$ for $t < t_j$, and the analogous argument works for $\beta = 1$ with the corresponding estimate. Next we show that $x$ solves equation (\ref{coarsening_equation}). Fix $k \in \Z$, $0 \leq j_1 < j_2$ and integrate equation (\ref{coarsening_equation}) with $x^{(n)}$ from $s_1$ to $s_2$, where $t_{j_1} \leq s_1 < s_2 \leq t_{j_2}$, to obtain

\begin{align}
x^{(n)}(s_2,k) - x^{(n)}(s_1,k) = \int_{s_1}^{s_2} \Delta_\sigma F_\beta \left(x^{(n)} \right) (t,k) \ \rd t.
\end{align}

\noindent By construction, the function $t \mapsto \sigma_\pm(x^{(n)}(t,.),k)$ is constant on $(s_1,s_2)$. Furthermore, by the construction of the sequence $x^{(n)}$, the number of values of $\sigma_\pm(x^{(n)}(t,.),k)$ regarded as a sequence in $n$ is finite, hence we can assume this to be independent of $n$ after taking another subsequence, in other words

\begin{align}
\sigma_\pm(x^{(n)}(t,.),k) = \sigma_\pm(x(t,.),k) \ \mathrm{for \ all } \ s_1 < t < s_2.
\end{align} 

\noindent Let $s_1 < t < s_2$. If $x(t,k) > 0$, then also $x^{(n)}(t,k) > 0$ for large $n$ and by the point-wise convergence of $x^{(n)}$ and the above identity we conclude

\begin{align}
\Delta_\sigma F_\beta(x^{(n)}) (t,k) \to \Delta_\sigma F_\beta(x(t,k)) \ \mathrm{as} \ n \to \infty.
\end{align}

\noindent On the other hand, $x(t,k) = 0$ implies $\eta_k \leq t_{j_1}$ and $x^{(n)}(t,k) = 0$, since $\eta^{(n)}_k = \eta_k$ for $n$ large, hence we also get $\Delta_\sigma F_\beta(x^{(n)}) (t,k) \to \Delta_\sigma F_\beta(x(t,k))$. Then we can apply dominated convergence, where we use that $x^{(n)}$ is locally bounded in the case $\beta > 0$ and the lower bound from Lemma \ref{Vanishing_time_estimate} in the case $\beta < 0$, and pass to the limit in the above integral identity to conclude

\begin{align}
x(s_2,k) - x(s_1,k) = \int_{s_1}^{s_2} \Delta_\sigma F_\beta(x)(t,k) \ \rd t,
\end{align}

\noindent which shows that $x$ is a solution to the coarsening equation.\newline

\noindent It remains to show that $x$ satisfies the desired bounds. We first consider the case $\beta \neq 1$. By (\ref{approximate_property_4}) we have

\begin{align}
x^{(n)} (t_j,k) \geq \frac{1}{2}\theta^{j},
\end{align} 

\noindent for all $j \in \N$ and living particles, and by convergence the same inequality holds in the limit $n \to \infty$. In particular we have

\begin{align}
\langle x(t_j) \rangle_{\sigma}^- \geq \frac{1}{2}\theta^{j}.
\end{align}

\noindent On the other hand, it is easy to check that $\langle x \rangle_{\sigma}^-$ is conserved between particle vanishings, hence we have

\begin{align}
\langle x(t) \rangle_{\sigma}^- \geq \frac{1}{2}\theta^{j},
\end{align}

\noindent whenever $t_{j} \leq t \leq t_{j+1}$. Because of (\ref{vanishing_time_formula}) we have

\begin{align}
t_j \sim \theta^{(1-\beta)j}.
\end{align}

\noindent This means that $t_{j} \leq t \leq t_{j+1}$ implies

\begin{align}
\frac{1}{1-\beta}\log_\theta(t) - C \leq j \leq \frac{1}{1-\beta}\log_\theta(t) + C,
\end{align}

\noindent and consequently

\begin{align}
\langle x(t) \rangle_{\sigma}^- \gtrsim t^\frac{1}{1-\beta}.
\end{align}

\noindent The upper bound on $||x||_\infty$ follows in the same way by (\ref{approximate_property_3}). For the case $\beta = 1$ the same argument applies, but in this case we have

\begin{align}
t_j = jT,
\end{align}

\noindent which leads to 

\begin{align}
\langle x(t) \rangle_{\sigma}^- \gtrsim \theta^{\frac{t}{T}} = \exp(\lambda t).
\end{align}

\noindent The restriction $\lambda \leq 2$ follows from the fact that equation (\ref{coarsening_equation}) for $\beta = 1$ implies $\dot{x} \leq 2x$.
\end{proof}

\noindent The proof of Corollary \ref{instability_result} follows easily with a very similar argument:

\begin{proof}[Proof of Corollary \ref{instability_result}]
Let $\varepsilon > 0$. It suffices to consider the case $c = 1/2$. We apply the same construction as in the proof of Theorem \ref{main_result} (with potentially different $\varepsilon$ in the definition of $\theta$) to get an approximate solution $x^{(n)}$, but in the iteration scheme we apply an additional step to $u^{(n,n)}$, satisfying $1/2 \leq u^{(n,n)} \leq 1 $ according to(\ref{iteration_estimate}). Using Lemma \ref{key_lemma} with $\varepsilon$ from above yields $\tilde{T}$, only depending on $\beta$ and $\varepsilon$ and a solution $u^{(n,n+1)}$ to the backward equation that satisfies 

\begin{align}
\left | u^{(n,n+1)}(\tau_{n,n} + \tilde{T}, \cdot) - \frac{1}{2} \right | \leq \varepsilon.
\end{align}

\noindent Then the sequence $x^{(n)}$ has the same properties as in the proof of Theorem \ref{main_result} with the addition that the initial data are in an $\varepsilon$-ball around $1/2$ for all $n$ by the above inequality, which gives the desired result after sending $n$ to infinity.

\end{proof}

\noindent \textbf{Remarks.} 

\begin{enumerate}
\item For $\beta \in (0,1)$ the achieved growth rate is optimal, because equation (\ref{coarsening_equation}) implies $\dot{x} \leq 2x^\beta$, which can be integrated to obtain $||x||_\infty \lesssim t^\frac{1}{1-\beta}$.
\item The fact that it was possible to choose $\tau_{n,j} - \tau_{n,j-1} = T\theta^{n+1-j}$ in the iteration step was crucial to obtain the desired growth rates. By comparison principle however, the estimate (\ref{iteration_estimate}) also remains valid if we choose a much larger time-span between particle insertions. This means that the above method can be adapted to produce solutions that are unbounded but grow arbitrarily slowly.
\item For convenience we chose $x^{(n)}_\mathrm{ter}$ to be constant in the approximation scheme. For the construction however we only used that $\frac{1}{2}\theta^n \leq x^{(n)}_\mathrm{ter} \leq \theta^n$ so that one can use arbitrary sequences satisfying these bounds in our construction scheme to produce solutions with the same coarsening behaviour.
\end{enumerate}

\newpage

\appendix

\section{Appendix}

Here we we address all technical results that were used in the previous sections and either prove them or give a reference. In the first three parts we discuss aspects of the discrete fast diffusion equation, while the rest of the appendix contains results about parabolic H\"older regularity in the discrete setting.

\subsection{The equation $\del_t u = \Delta G_\beta(u)$}

\noindent We consider the Cauchy problem for the discrete fast diffusion equation 

\begin{align} \label{appendix_backward_equation}
\begin{cases}
\del_t u = \frac{\beta}{|\beta|} \Delta u^\beta =  \Delta G_\beta(u) \quad \mathrm{in} \ (0,\infty) \times \Z,  \\
u(0,\cdot) = u_0,
\end{cases}
\end{align}

\noindent with $\beta \in (-\infty,0) \cup (0,1]$, $u_0 \in \ell_+^\infty(\Z)$ and 

\begin{align}
\Delta u = u(k-1) -2u(k) + u(k+1).
\end{align}

\noindent We are concerned with the long-time existence of classical solutions:

\begin{defn}
A function $u: [0,\infty) \to \ell_+^\infty(\Z)$ is a solution to problem (\ref{appendix_backward_equation}) if the following conditions are satisfied:

\begin{enumerate}
\item $t \mapsto u(t,\cdot)$ is in $C^0([0,\infty);\ell_+^\infty(\Z))$ and $u(0,\cdot) = u_0$.

\item For every $k \in \Z$ we have $u(.,k) \in C^1((0,\infty);\R_{> 0})$ and  

\begin{align}
\frac{\rd}{\rd t} u(t,k) = \Delta G_\beta(u)(t,k),
\end{align}

\noindent for all $k \in \Z$ and $t > 0$.
\end{enumerate}

\end{defn}

\noindent For positive $\beta$ it is not hard to prove the existence of a solution for any kind of initial data, since $G_\beta$ is bounded at zero and one has simple a-priori estimates due to the comparison principle (see below). For negative $\beta$ the existence of a sufficiently regular solution for arbitrary data cannot be expected due to $G_\beta(x)$ becoming singular at $x = 0$. Similar to the result in \cite{HelmersNV2016}, we have to restrict to initial data which satisfy a certain positivity condition:

\begin{defn}
For $u \in \ell_+^\infty(\Z)$ and $d > 0$ let 

\begin{align}
\sigma_+(u,k,d) &= \inf \{l > k: u(l) \geq d \}. 
\end{align}

\noindent Then for any $L \in \N$ and $d > 0$ we define

\begin{align}
\mathcal{P}_{L,d} = \left \{u \in \ell_+^\infty: \ \sup_{k \in \Z} \sigma_+(u,k,d) \leq L \right \}.
\end{align}

\end{defn}

\noindent In other words, $u \in \mathcal{P}_{L,d}$ means that particles with large mass cannot be very far apart. This is also relevant for the case $ \beta > 0$ since it allows us to prove certain Harnack-type positivity estimates. We have the following result: 

\begin{thm}\label{appendix_existence_thm}
Let $\beta \in (-\infty,0) \cup (0,1]$, and consider initial data $u_0 \in \mathcal{P}_{L,d}$. Then the following statements hold:
\begin{enumerate}

\item If $\beta \in (-\infty,0) \cup (0,1)$, there exists a positive constant $c = c \left(L, d, \beta, ||u_0||_\infty \right)$ and a solution $u$ to equation (\ref{appendix_backward_equation}) on $[0,\infty)$ with initial data $u_0$ satisfying

\begin{align}
u(t,k) &\geq c \left( 1 \wedge t^{\frac{1}{1-\beta}} \right), \ \mathrm{for} \ \mathrm{all} \ k \in \Z. \label{appendix_positivity_estimate}
\end{align}

\item If $\beta = 1$, the same statement holds with estimate (\ref{appendix_positivity_estimate}) replaced by

\begin{align}
u(t,k) &\geq c \exp(-2t)I_L(2t), 
\end{align}

\noindent where $I_L(t)$ denotes the $L$-th modified Bessel function of the first kind.

\item Comparison principle: If $ c_1 \leq u_0 \leq c_2$, then $u$ satisfies these bounds for all times.

\end{enumerate}
\end{thm}

\noindent In this and the next two sections we give a full proof of the above result. The general strategy to prove existence of solutions for equation (\ref{backward_equation}) is to use regularization and standard ODE theory. Instead of infinitely many particles with non-negative mass we first consider a periodic $N$-particle ensemble where particles have strictly positive mass. The first important a-priori estimate is the comparison principle:

\begin{lem}[Finite positive ensemble]
Let $\mathbb{T}_N$ denote the one-dimensional periodic lattice with $N$ lattice points. Let $u_0 \in \ell_+^\infty(\mathbb{T}_N)$ with $ 0 < \delta \leq u_0 \leq C$. Then there exists a unique solution $u:[0,\infty) \to \ell_+^\infty(\mathbb{T}_N)$ of (\ref{appendix_backward_equation}) with $\delta \leq  u(t,\cdot) \leq C$.
\end{lem}

\begin{proof}
The proof is very similar to Lemma 2 in \cite{EsedogluG2009} and a standard maximum principle argument. Because $u_0 \geq \delta$, standard ODE theory gives the existence and uniqueness of a smooth solution $u$ on the time interval $[0,t^*]$ to equation (\ref{appendix_backward_equation}) with $\delta /2 \leq u(t,.) \leq 2C$ for some positive $t^*$. For small $\varepsilon > 0$ we then consider the solution $u_\varepsilon$ of the modified problem

\begin{align}
\del_t u_\varepsilon = \Delta G_\beta(u_\varepsilon) + \varepsilon, \\
u_\varepsilon(0,\cdot) = u_0,
\end{align}

\noindent that exists on the same time interval as $u$ and satisfies the same bounds after possibly making $t^*$ smaller. Because $G_\beta$ is smooth on $[\delta/2,2C]$ we have that $u_\varepsilon \to u$ uniformly on $[0,t^*]$. We claim that $u_\varepsilon$ attains its minimum over $[0,t^*] \times \mathbb{T}_N$ at $t = 0$. If not, there exists $t_0 \in (0,t^*]$ and $k_1$ such that $u_\varepsilon(t_1,k_1)$ is the absolute minimum. Consequently we get 

\begin{align}
0 \geq \del_t u_\varepsilon(t_1,k_1) = \Delta G_\beta(u_\varepsilon)(t_1,k_1) + \varepsilon \geq \varepsilon,
\end{align}

\noindent a contradiction. Here we used that $G_\beta$ is increasing. Hence $u_\varepsilon(t,.) \geq \delta$ for all $t \in [0,t^*]$ and, sending $\varepsilon \to 0$, the same bound holds for $u$. The same argument for the maximum where $+\varepsilon$ is replaced with $-\varepsilon$ yields that $u \leq C$. Iterating from $t = t^*$, we see that the solution can be extended to $[0,\infty)$ and always satisfies the desired bounds.

\end{proof}

\noindent From this result we can easily pass to the limit as $N \to \infty$ to obtain solutions for infinite numbers of particles:

\begin{cor}[Infinite positive ensemble] \label{comparison_principle_lemma}
Let $u_0 \in \ell_+^\infty(\Z)$ with $ 0 < \delta \leq u_0 \leq C$. Then there exists a solution $u:[0,\infty) \to \ell_+^\infty(\Z)$ of (\ref{appendix_backward_equation}) with $\delta \leq  u(t,.) \leq C$.
\end{cor}

\begin{proof}
This is a standard compactness argument. We choose $u_0^{(N)}$ to be $N$-periodic such that $u_0^{(N)}(k) \to u_0(k)$ for each $k \in \Z$. Let $u^{(N)}$ be the corresponding solutions from the above lemma. Then due to the a-priori bounds $\delta \leq  u^{(N)} \leq C$ and equation (\ref{appendix_backward_equation}) we have

\begin{align}
\left|\frac{\rd}{\rd t} u^{(N)} \right| = \left |\Delta G_\beta \left(u^{(N)} \right) \right| \leq K(\delta,C),
\end{align}

\noindent hence the solutions are uniformly Lipschitz continuous. Applying the Arzela-Ascoli Theorem and a diagonal argument we can extract a convergent subsequence (not relabeled) such that $u^{(N)}(\cdot,k) \to u(\cdot,k)$ for some uniformly on compact time intervals, where $u(\cdot,k) \in C^0([0,\infty)$. In particular $u$ satisfies the same bounds as $u^{(N)}$. Integrating (\ref{appendix_backward_equation}) in time and passing to the limit (which is possible due to the a-priori bounds) then yields 

\begin{align}
u(t,k) = u_0(k) + \int_{0}^{t} \Delta_\sigma G_\beta(u)(s,k) \ \rd s.
\end{align}

\noindent This in turn shows that $u(.,k)$ is continuously differentiable and solves (\ref{appendix_backward_equation}) pointwise. Again, the bounds on $u$ yield Lipschitz continuity in $\ell_+^\infty(\Z)$.
\end{proof}

\noindent For our purpose, we need the existence of solutions in particular for initial data with mass zero particles. The general strategy is to approximate the initial data by regularized data via

\begin{align}
u_{0,\delta} = u_0 \vee \delta.
\end{align}

\noindent The above existence result then yields long-time solutions $u_\delta$ with initial data $u_{0,\delta}$. In the case $\beta > 0$ one can pass to the limit $\delta \to 0$ in the same manner as above, since $G_\beta$  is bounded at zero, yielding a general existence result:

\begin{cor}[Existence for positive $\beta$] 
Let $\beta \in (0,1]$ and $u_0 \in \ell_+^\infty(\Z)$. Then there exists a solution $u:[0,\infty) \to \ell_+^\infty(\Z)$ of (\ref{appendix_backward_equation}).
\end{cor}

\noindent Alternatively it is likely possible to prove this result directly via an infinite dimensional fixed-point method. Since we need Corollary (\ref{comparison_principle_lemma}) for the case $\beta < 0$ anyway, the above method is the fastest for our purpose. In the next section we deal with the negative $\beta$ case, including existence and the positivity estimate (\ref{appendix_positivity_estimate}). Then we prove the positivity estimate for positive $\beta$, completing the proof of Theorem \ref{appendix_existence_thm}. \newline

\subsection{Existence of solutions for $\beta < 0$}

In the following we always assume $\beta < 0$. The key idea to prove existence of solutions is to exploit the fact that regions which are enclosed by large particles (called \textit{traps}) are screened from the rest of the particles, very similar to \cite{HelmersNV2016}. One important difference however is the fact that the backward equation does not yield a-priori estimates on the persistence of traps. We make the following definition:

\begin{defn}
We say that a solution $u$ to equation (\ref{appendix_backward_equation}) with initial data $u_0 \in  \mathcal{P}_{L,d}$ has the persistence property on $[0,T]$ if $u(0,k) \geq d$ implies $u(t,k) \geq \frac{d}{2}$ for all $t \in [0,T]$.
\end{defn}

\noindent By making use of the theory for the coarsening equation developed in \cite{HelmersNV2016} we have the following result concerning H\"older regularity:

\begin{lem}\label{trap_hoelder_estimate}
There exist constants $T' = T'(\beta,d)$ and $C = C(\beta,L) > 0$ such that the following holds: If a solution $u$ to equation (\ref{appendix_backward_equation}) with initial data $u_0 \in  \mathcal{P}_{L,d}$ has the persistence property on $[0,T]$ and $T \leq T'$, then 

\begin{align}
|u(t_2,k)-u(t_1,k)| &\leq C|t_2 - t_1|^{\frac{1}{1-\beta}},
\end{align}

\noindent for all $t_2,t_1 \in [0,T]$ and $k \in \Z$.
\end{lem}

\begin{proof}
We consider the time reversed function

\begin{align}
x(s,k) = u(T - s,k), 
\end{align}

\noindent then $x$ is a solution to the coarsening equation for $ 0 \leq s \leq T$ that satisfies $x(0,.) \in \mathcal{P}_{L,\frac{d}{2}}$.  Applying Lemma 3.3 from \cite{HelmersNV2016} (if $T \leq T^*(\beta,\frac{d}{2}) =: T'$) yields the desired H\"older continuity for $x$, and thus also for $u$.

\end{proof}

\noindent From this result we derive the first a-priori estimate:

\begin{lem} \label{a_priori_traps}
Let $u_{0,\delta}$ be as above and let $u_\delta$ be the corresponding solution of equation (\ref{appendix_backward_equation}), which exists by Lemma \ref{comparison_principle_lemma}. Then there exists $T = T(L,d) > 0$ such that $u_\delta$ has the persistence property on $[0,T]$.
\end{lem}

\begin{proof}

 First we note that because the solution satisfies $u_\delta \geq \delta$ for all times then we have the Lipschitz estimate

\begin{align}
|\del_t u_\delta| \leq 4\delta^{\beta}.
\end{align}

\noindent This means that $u_\delta(0,k) \geq d$ implies $u_\delta(t,k) \geq \frac{d}{2}$ for $0 \leq t \leq t_0$, where 

\begin{align}
t_0 = \frac{d}{8\delta^\beta}.
\end{align}

\noindent Let $T$ be the largest time such that $u_\delta(0,k) \geq d$ implies $u_\delta(t,k) \geq \frac{d}{2}$ on $[0,T]$. By the above considerations we already know that $T > 0$. If $T \leq T'(\beta,d)$ we can apply Lemma \ref{trap_hoelder_estimate} and get

\begin{align}
u_\delta(t,k) \geq u_\delta(0,k) - Ct^{\frac{1}{1-\beta}}.
\end{align}

\noindent  If $u_\delta(0,k) \geq d$, this implies $u_\delta(T,k) \geq d - CT^\frac{1}{1-\beta}$. On the other hand, by the definition of $T$ there exists such a $k$ with $u_\delta(T,k) \leq \frac{3d}{4}$, hence

\begin{align}
\frac{3d}{4} \geq d - C(L)T^\frac{1}{1-\beta},
\end{align} 

\noindent which gives a lower bound for $T$ in terms of $L$ and $d$.

\end{proof}

\noindent The next a-priori estimate is crucial to get uniform H\"older bounds on $u_\delta$, as well as integral bounds which are needed to pass to the limit. 

\begin{lem} \label{a_priori_lower_bound}
Let $u_\delta$ be as above. Then there exists $c = c(\beta,L,d) > 0$ and $t^* = t^*(\beta,L,d) > 0$ such that 

\begin{align}
u_\delta(t,\cdot) \geq ct^{\frac{1}{1-\beta}},
\end{align}

\noindent for $0 \leq t \leq t^*$.
\end{lem} 

\begin{proof}
We apply a very similar argument as in the proof of Lemma 3.5 in \cite{HelmersNV2016}. If the statement is false, there exist sequences $u^{(n)}$, $t_n \to 0$, $\delta_n \to 0$ and $k_n \in \Z$ such that

\begin{align}
u^{(n)}_{\delta_n}(t_n,k_n) \leq \frac{1}{n} t_n^{\frac{1}{1-\beta}}.
\end{align}

\noindent By translation invariance we can assume that $k_n = k_0$ is constant. We rescale and define

\begin{align}
v_n(s,k) = t_n^{\frac{1}{\beta-1}}u^{(n)}_{\delta_n}(t_ns,k).
\end{align}

\noindent Then $v_n$ is a solution to equation (\ref{appendix_backward_equation}) with $v_n(1,k_0) \to 0$. Additionally we have $v_n(0,.) \in \mathcal{P}_{L,d}$ and $v_n$ satisfies the persistence property on $[0,1]$ for large $n$ by Lemma \ref{a_priori_traps}. Since $t_n^{\frac{1}{\beta -1}} d \to \infty$ and $T' \to \infty$ as $d \to \infty$ we also have that $v_n$ is uniformly H\"older continuous by Lemma \ref{trap_hoelder_estimate}. Let $B$ be the largest set of consecutive indices containing $k_0$ such that 

\begin{align}
\liminf_{n \to \infty} v_n(1,k) > 0,
\end{align}

\noindent for $k \in B$. Observe that we have $|B| \leq L$ due to $t_n^{\frac{1}{\beta -1}} d \to \infty$ and the persistence property. Let $l_-, l_+$ be the nearest particle index to the left, respectively to the right of $B$.  We restrict to a subsequence such that $v_n(1,k) \to 0$ for $k \in B$ and $v_n(1,l_\pm) \geq \lambda > 0$. If we define the local mass $M_n(s)$ as

\begin{align}
M_n(s) = \sum_{k \in B}v_n(s,k),
\end{align}

\noindent then an elementary calculation gives

\begin{align} \label{rate_of_change_local_mass}
\frac{\rd}{\rd s}M_n(s) = v_n^\beta(s,l_- + 1) - v_n^\beta(s,l_-) + v_n^\beta(s,l_+ -1) - v_n^\beta(s,l_+).
\end{align}

\noindent Due to uniform H\"older continuity and particles in $B$ going to zero there exists $\varepsilon > 0$ such that 

\begin{align}
v_n^\beta(s,l_\pm) \leq \frac{1}{2} v_n^\beta(s,l_\pm \mp 1),
\end{align}

\noindent for $s \in [1-\varepsilon,1]$ and large enough $n$. Using equation (\ref{rate_of_change_local_mass}) on this time interval we obtain

\begin{align}
2\frac{\rd}{\rd s}M_n(s) \geq v_n^\beta(s,l_- + 1) + v_n^\beta(s,l_+ - 1) \geq M_n^{\beta}(s),
\end{align}

\noindent which, after integrating from $1-\varepsilon $ to $1$ yields

\begin{align}
M_n(1) \geq \tilde{\varepsilon} > 0,
\end{align}

\noindent which gives a contradiction after sending $n$ to infinity.

\end{proof}

\noindent This result gives us the a-priori estimates we need to pass to the limit:

\begin{cor}[H\"older continuity] \label{a_priori_hoelder_estimate}
Let $u_\delta$ and $t^*$ be as above. Then there exists $C = C(\beta,L,d)$ such that

\begin{align}
|u_\delta(t_2,k)-u_\delta(t_1,k)| &\leq C|t_2 - t_1|^{\frac{1}{1-\beta}},
\end{align}

\noindent for all $t_2,t_1 \in [0,t^*]$ and $k \in \Z$.

\end{cor}

\begin{proof}
Let $t_2,t_1 \in [0,t^*]$, $t_2 > t_1$. We integrate (\ref{appendix_backward_equation}) in time and estimate

\begin{align}
|u(t_2,k)-u(t_1,k)| &\leq \int_0^t |\Delta G(u_\delta)(s,k)| \ \rd s \lesssim \int_{t_1}^{t_2} s^{\frac{\beta}{1-\beta}} \ \rd s \\
&\sim t_2^{\frac{1}{1-\beta}} - t_1^{\frac{1}{1-\beta}} \leq (t_2-t_1)^\frac{1}{1-\beta},
\end{align}

\noindent where we used Lemma \ref{a_priori_lower_bound} to estimate $|\Delta G_\beta(u_\delta)|$.
\end{proof}

\noindent With these preparations we can prove the first statement of Theorem \ref{appendix_existence_thm} for negative $\beta$:

\begin{proof}[Proof of existence and positivity bound for $\beta < 0$]

Let $u_{0,\delta}$ and $u_\delta$ as above. By Corollary \ref{a_priori_hoelder_estimate} and the Arzela-Ascoli Theorem there exists a subsequence $\delta \to 0$ such that $u_\delta \to u$ uniformly on $[0,t^*]$. Moreover, by Lemma \ref{a_priori_lower_bound} we have

\begin{align}
u_\delta^\beta(t,k) \leq c^\beta t^{\frac{\beta}{1-\beta}}, 
\end{align}

\noindent which implies $u_\delta^\beta \to u^\beta$ in $L^1([0,t^*])$. Using this we can pass to the limit in the integral equation 

\begin{align}
u_\delta(t,k) = u_{0,\delta}(k) - \int_0^t \Delta G(u_\delta)(s,k) \ \rd s,
\end{align}

\noindent showing that $u$ is a solution to the backward equation with initial data $u_0$ on $[0,t^*]$. Since the lower bound from Lemma \ref{a_priori_lower_bound} also holds in the limit, we can extend the solution from $t^*$ to arbitrary times via comparison principle, which also changes the lower bound to $\sim 1 \wedge t^\frac{1}{1-\beta}$ for large times.

\end{proof}

\subsection{Harnack-type inequality for $0 < \beta \leq 1$}

In the previous part Lemma \ref{a_priori_lower_bound} was crucial to prove existence of a solution. The result of the lemma, together with the positivity condition $\mathcal{P}_{L,d}$ can be interpreted as a Harnack-type inequality, see \cite{BonforteV2006}. For $ 0 < \beta < 1 $ a similar result holds, the equation however behaves differently and the indirect proof does not work here. We will pursue another approach and show the inequality directly with explicit constants, handling the case $\beta = 1$ separately. The key observation is that a large particle next to a small particle will always induce growth on the small particle, despite the size of the other neighbour of the small particle. This decouples the equation in a sense and we only need to study the local problem:

\begin{lem}[Local Problem]\label{harnack_iteration_step_lemma}
Let $0 < \beta < 1$ and $T > 0$. Consider two functions $F \in C^0([0,T];[0,\infty))$ and $u \in C^1([0,T];[0,\infty))$ which satisfy

\begin{align}
F(t) &\geq ct^{\frac{\beta}{1-\beta}}, \\
\dot{u}(t) &\geq F(t) - 2u^\beta(t).
\end{align}

\noindent on $[0,T]$. Then $u$ satisfies 

\begin{align}
u(t) \geq \eta_1(c)t^{\frac{1}{1-\beta}},
\end{align}

\noindent on the interval $[0,T]$, where $\eta_1$ is a positive strictly increasing function which depends only on $\beta$.
\end{lem}

\begin{proof}
We define the rescaled function

\begin{align}
v(t) = t^{\frac{1}{\beta -1}}u(t),
\end{align}

\noindent on the half-open interval $(0,T]$. Then it suffices to show that $v$ is bounded from below by $\eta_1(c)$. We calculate

\begin{align}
t\dot{v}(t) &= t\left(t^{\frac{1}{\beta -1}}\dot{u}(t) + \frac{1}{\beta - 1}t^{\frac{1}{\beta -1}-1}u(t) \right) \\
&\geq t\left(t^{\frac{1}{\beta -1}}F(t) -2t^{\frac{1}{\beta -1}}u^\beta(t) + \frac{1}{\beta - 1}t^{\frac{1}{\beta -1}-1}u(t) \right)  \\
&= t^{\frac{\beta}{\beta -1}}F(t) - 2v^\beta(t) - \frac{1}{1-\beta}v(t) \\
&=: t^{\frac{\beta}{\beta -1}}F(t) - \theta(v(t)),
\end{align}

\noindent in particular the assumption on $F$ implies that 

\begin{align} \label{harnack_rescaled_differential_inequality}
t\dot{v}(t) \geq c - \theta(v(t)). 
\end{align}

\noindent Since the function $\theta$ is strictly increasing on $[0,\infty)$ we can define the inverse function $\eta_1 = \theta^{-1}$, which is also strictly increasing. We claim that $v \geq \eta_1(c)$ on $(0,T]$. If this is not true, there is $\varepsilon > 0$ and $t^* \in (0,T]$ such that $v(t^*) \leq \eta_1(c) - \varepsilon$. In particular we have

\begin{align}
c - \theta(v(t^*)) \geq \tilde{\varepsilon} > 0,
\end{align} 

\noindent for some $\tilde{\varepsilon} > 0$. But then the differential inequality (\ref{harnack_rescaled_differential_inequality}) implies that $v(t) \leq \eta_1(c) - \varepsilon$, and hence $c - \theta(v(t)) \geq \tilde{\varepsilon}$ for all $ t \in (0,t^*]$. Dividing by $t$ and integrating (\ref{harnack_rescaled_differential_inequality}) in time gives

\begin{align}
v(t^*) - v(t) \geq \varepsilon \log \left(\frac{t^*}{t} \right),
\end{align}

\noindent for all $0 < t \leq t^*$. Sending $t$ to zero then gives a contradiction.

\end{proof}

\noindent The above lemma enables us to prove a Harnack type inequality:

\begin{lem}[Harnack type inequality] \label{harnack_inequality_positive}
Let $u$ be a solution to equation (\ref{appendix_backward_equation}) with $0 < \beta < 1$ and initial data $0 \leq u_0 \leq 1$. Then we have 

\begin{align}
u(t,k) \geq \eta(|k-l|)(t-s)^{\frac{1}{1-\beta}},
\end{align}

\noindent for all $k,l \in \Z$ and $0 \leq t -s \leq t^*(u(s,l))$. The function $\eta$ is strictly positive and the function $t^*$ is non-negative, strictly increasing with $t^*(u) = 0$ iff $u=0$. Furthermore, both functions depend only on $\beta$.

\end{lem}

\begin{proof}
Due to translation invariance in space and time it suffices to consider the case $s=0$ and $l=0$. We will make an iterative argument, using Lemma \ref{harnack_iteration_step_lemma} in each step. First we note that due to $u_0 \leq 1$ and the comparison principle we have the Lipschitz estimate

\begin{align}
|\dot{u}| \leq 4,
\end{align}

\noindent in particular

\begin{align}
u(t,0) \geq u_0(0) - 4t.
\end{align}

\noindent The case $u_0(0) = 0$ is trivial. If $u_0(0) > 0$, the Lipschitz estimate implies

\begin{align}
u(t,0) \geq t^{\frac{1}{1-\beta}},
\end{align}

\noindent whenever $u_0(0) - 4t - t^{\frac{1}{1-\beta}} > 0$. If we set $f(t) = 4t + t^{\frac{1}{1-\beta}}$ then $t^*$ is defined as the inverse of $f$. Thus the above lower bound holds for $0 \leq t \leq t^*(u_0(0))$. For $u(t,1)$ we have

\begin{align}
\dot{u}(t,1) &= u^\beta(t,0) - 2u^\beta(t,1) + u^\beta(t,2) \\
&\geq u^\beta(t,0) - 2u^\beta(t,1).
\end{align}

\noindent This means that $u^\beta(t,0)$ and $u(t,1)$ satisfy the assumptions of Lemma \ref{harnack_iteration_step_lemma} with $T = t^*(u_0(0))$ and $c = 1$. Thus we have

\begin{align}
u(t,1) \geq \eta_1(1)t^{\frac{1}{1-\beta}},
\end{align}

\noindent on $[0,t^*(u_0(0))]$. Now we can successively apply the same argument to the pairs of functions $(u^\beta(t,1)$,$u(t,2))$,...,$(u^\beta(t,k-1)$,$u(t,k))$, for $k \in \N$. The argument for $-k$ is the same. Then the desired inequality follows with the function $\eta = \eta(r)$ ($r \in \N$) defined as

\begin{align}
\eta(r) &= \eta_1^{(r)}(1), \\
\eta(0) &= 1,
\end{align} 

\noindent where $\eta_1^{(r)}$ means that $\eta$ is $r$ times composed with itself.

\end{proof}

\begin{lem}\label{harnack_inequality_linear}
Let $u$ be a solution to the constant coefficient linear equation equation (\ref{appendix_backward_equation}) with $\beta = 1$. Then for every $k \in \Z$ and $N \in \mathbb{N}$ we have

\begin{align}
u(t,k) \geq M(u_0,k,N) \exp(-2t)I_N(2t),
\end{align}

\noindent where $I_k(t)$ denotes the $k$-th modified Bessel function of the first kind and 

\begin{align}
M(u_0,k,N) = \sum_{l=-N}^N u_0(k-l)
\end{align}

\noindent denotes the local initial mass. 

\end{lem}

\begin{proof}
In the linear constant-coefficient case we can give an explicit formula by Fourier-analysis: We write

\begin{align}
\hat{u}(t,\theta) = \sum_{k=-\infty}^{k=+\infty} u(t,k)\exp(-ik\theta),
\end{align}

\noindent taking the time derivative on both sides and using the equation then yields

\begin{align}
\del_t \hat{u}(t,\theta) &= \exp(-i\theta)\sum_{k=-\infty}^{k=+\infty} u(t,k)\exp(-ik\theta) \\
&- 2\sum_{k=-\infty}^{k=+\infty} u(t,k)\exp(-ik\theta) \\
&+ \exp(i\theta)\sum_{k=-\infty}^{k=+\infty} u(t,k)\exp(-ik\theta) \\
&= 2(\cos(\theta) - 1)\hat{u}(t,\theta).
\end{align}

\noindent We solve this ODE in $t$ with initial data $\hat{f}$ to obtain

\begin{align}
\hat{u}(t,\theta) = \hat{f}(\theta)\exp(2t(\cos(\theta) - 1)),
\end{align}

\noindent which gives the discrete heat kernel

\begin{align}
\phi(t,k) &= \exp(-2t) \frac{1}{2\pi} \int_{-\pi}^\pi \exp(2t\cos(\theta)-ik\theta) \ \rd \theta \\
&= \exp(-2t) \frac{1}{\pi} \int_{0}^\pi \exp(2t\cos(\theta))\cos(k\theta) \ \rd \theta \\
&= \exp(-2t)I_k(2t),
\end{align}

\noindent where $I_k$ is the $k$th modified Bessel function of the first kind. Then the desired inequality follows directly by the standard representation

\begin{align}
u(t,k) = \sum_{l \in Z}u_0(k-l)\phi(t,l),
\end{align}

\noindent the fact that $\phi$ is decreasing in the second argument and the obvious estimate.
\end{proof}

\noindent We summarize the findings of this section and prove the remaining statements of Theorem \ref{appendix_existence_thm}:

\begin{proof}[Proof of positivity estimate for $0 < \beta \leq 1$]
First we consider the case $\beta \neq 1$. It suffices to consider the case $u_0 \leq 1$ by scaling (this means that for general data the constants get an additional dependence on $||u_0||_\infty$). Because $u_0 \in \mathcal{P}_{L,d}$, for every $k \in \Z$ there exists $k'$ with $u_0(k') \geq d$ and $|k-k'| \leq L$. Then Lemma \ref{harnack_inequality_positive} with $s=0$ and $l = k'$ yields 

\begin{align}
u(t,k) \geq \min_{j=1,..,L}\eta(j) t^\frac{1}{1-\beta},
\end{align}

\noindent for $0 \leq t \leq t^*(d)$ because $t^*$ is monotone. For $\beta = 1$ we note that $u_0 \in \mathcal{P}_{L,d}$ implies that $M(u,k,L) \geq 2d$, since there are at least two terms in the sum that are greater than or equal to $d$ by definition of $\mathcal{P}_{L,d}$. Then the statement follows directly from Lemma \ref{harnack_inequality_linear}.
\end{proof} 

\subsection{Nash-Aronson estimates and H\"older continuity}

For $u = u(k) \in \ell^\infty(\Z)$ we define the forward and backward difference operators

\begin{align}
\del^+ u(k) &= u(k+1) - u(k) \\
\del^- u(k) &= u(k) - u(k-1).
\end{align}

\noindent For $a = a(t,k)$ with $0 < c_1 \leq a \leq c_2$ and $a(.,k) \in C^0([0,\infty))$ we consider the discrete analogue to a parabolic evolution equation in divergence form:

\begin{align} \label{discrete_parabolic_equation_appendix}
\begin{cases}
\del_t u = \del^-(a \del^+ u) =: \mathcal{L}(t)u \\
u(0,.) = u_0.
\end{cases}
\end{align}

\noindent We denote by $\phi(t,k,s,l)$ the fundamental solution to (\ref{discrete_parabolic_equation_appendix}), in other words, $\phi(.,.,s,l)$ is the solution to the above equation starting at time $s$ with initial data $\phi_0(k) = \delta_{kl}$. Since $\mathcal{L}(t)$ is a bounded operator from $\ell^2(\Z)$ to $\ell^2(\Z)$, $\phi$ can be written as

\begin{align}
\phi(t,k,s,l) = \left \langle \exp \left(\int_s^t \mathcal{L}(r) \ \rd r \right)\delta_l, \delta_k \right \rangle,
\end{align}

\noindent where $\delta_k$ are the canonical basis vectors in $\ell^2(\Z)$. The general solution starting at time $t = s$ to (\ref{discrete_parabolic_equation}) is then given by 

\begin{align}
u(t,k) = \sum_{l \in \Z^d}\phi(t,k,s,l)u_0(l).
\end{align} 

\noindent We also define the reduced fundamental solution

\begin{align}
\psi(t,k) &= \phi(t,k,0,0) = \phi(t,0,0,k), 
\end{align}

\noindent and the corresponding "macroscopic" rescaled function

\begin{align}
&U:\R \to [0,+\infty) \\
&U(t,\xi) = t^\frac{1}{2}\psi(t,\lfloor t^\frac{1}{2} \xi \rfloor).
\end{align}

\noindent We have the following Nash-Aronson estimates on the fundamental solution:\newline

\begin{thm}\label{discrete_H_continuity_thm}
There exist constants $t_0 > 0$, $C > 0$ and $\alpha > 0$, depending only on the bounds on $a$, such that the following statements hold:

\begin{itemize}
\item Aronson estimate:
\begin{align}
\psi(t,k) &\leq \frac{C}{1 \vee t^\frac{1}{2}}\exp \left(-\frac{|k|}{1 \vee t^\frac{1}{2}} \right), \label{Aronson_estimate}
\end{align}

\noindent for every $k \in \Z$ and $t \geq 0$.

\item Nash continuity estimate:
\begin{align}
|\psi(t,k) - \psi(t,l)| &\leq \frac{C}{t^\frac{1}{2}}\left(\frac{|k-l|}{t^\frac{1}{2}} \right)^\alpha, \label{Nash_estimate}
\end{align}

\noindent for every $k,l \in \Z$ and $t \geq 0$.

\end{itemize}

\end{thm}

\begin{proof}
Here we cite the results from Appendix B of \cite{GiacominOS2001}. Inequality (\ref{Aronson_estimate}) is precisely the statement of Proposition B.3. For the second inequality (\ref{Nash_estimate}) we first note that (\ref{Aronson_estimate}) implies $|\psi_{1/2}| \lesssim t^{-\frac{1}{2}}$. Then the desired estimate at a time $t^*$ follows from Proposition B.6 applied at $t = s = t^*/2$ with $f = \psi(t^*/2,.)$ and the semigroup property.
\end{proof}

\noindent These estimates have important consequences for the function $U$. Inequality (\ref{Aronson_estimate}) implies that 

\begin{align}
U(t,\xi) \leq \Phi(\xi),
\end{align}

\noindent for some integrable function $\Phi$. In particular the function family $U(t,.)$ are tight probability measures. On the other hand, the estimate (\ref{Nash_estimate}) implies that the function $U(t,.)$, which is a step-function by definition, becomes H\"older continuous in the following sense:

\begin{defn}[Approximate H\"older Continuity]
Let $\{f_n \} \subset L^\infty(\R)$ be a sequence of functions. Then $\{f_n \}$ is said to be \textbf{approximately H\"older continuous} with exponent $\alpha \in (0,1]$ if for every $\varepsilon > 0$ there exists $n = n(\varepsilon)$ such that $|x-y| \geq \varepsilon$ implies

\begin{align}
|f_n(x) - f_n(y)| \leq C |x-y|^\alpha, \label{approximate_H_continutiy}
\end{align}

\noindent for $n \geq n(\varepsilon)$ and a universal positive constant $C$. 
\end{defn}

\noindent The important observation is that H\"older continuity on the discrete microscopic level implies approximate H\"older continuity on the macroscopic scale:

\begin{lem}\label{fund_sol_H_continutiy_lemma}
The function $U(t,.)$ is approximately H\"older continuous as $t \to \infty$. Furthermore, the constants $C, \alpha$ and $t = t(\varepsilon)$ only depend on the bounds of the coefficient $a$ in (\ref{discrete_parabolic_equation_appendix}).
\end{lem}

\begin{proof}
By the estimate (\ref{Nash_estimate}) from Theorem \ref{discrete_H_continuity_thm} we calculate

\begin{align}
|U(t,\xi)-U(t,\eta)| \lesssim \left( \frac{|\lfloor t^\frac{1}{2} \xi \rfloor - \lfloor t^\frac{1}{2} \eta \rfloor |}{t^\frac{1}{2}} \right)^\alpha.
\end{align}

\noindent Since 

\begin{align}
\frac{|\lfloor t^\frac{1}{2} \xi \rfloor - \lfloor t^\frac{1}{2} \eta \rfloor |}{t^\frac{1}{2}} = |\xi - \eta| + \mathcal{O}(t^{-\frac{1}{2}}),
\end{align}

\noindent we get the desired estimate for $|\xi - \eta| \gtrsim t^{-\frac{1}{2}} $.

\end{proof}

\noindent The next result is of crucial importance for the main result of the paper. Denote by $\mathcal{T}_\delta(\R)$ the set of step-functions with step-width at least $\delta$. Then we have:

\begin{lem}\label{stepfct_approximation_lemma}
Let $(f_n) \subset L^1(\R)$ be tight and approximately H\"older continuous. Then for every $\varepsilon > 0$ there exists $n_0$ and $\delta > 0$, such that for every $f_n$ with $n \geq n_0$ there exists $\chi \in \mathcal{T}_\delta(\R)$ with

\begin{align}
||f_n - \chi||_{L^1(\R)} \leq \varepsilon.
\end{align}
\end{lem}

\begin{proof}
Let $\varepsilon > 0$. Because $(f_n)$ is tight in $L^1$ there exists $R > 0$ such that

\begin{align}
\int_{|x| \geq R} |f_n(x)| \ \rd x \leq \varepsilon,
\end{align}

\noindent hence it suffices to approximate $(f_n)$ in $L^\infty$. By approximate H\"older continuity there exists $n_0$ such that 

\begin{align}
|f_n(x) - f_n(y)| \leq C |x-y|^\alpha,
\end{align}

\noindent for $n \geq n_0$ and $|x-y| \geq R^{-\frac{1}{\alpha}}\varepsilon$. This means that the piecewise-constant interpolation $\chi$ of $f_n$ with step-width $R^{-\frac{1}{\alpha}}\varepsilon$ approximates $f_n$ uniformly up to an error of $R^{-1}\varepsilon^\alpha$, hence

\begin{align}
||f_n - \chi||_{L^1} \lesssim \varepsilon + \varepsilon^\alpha.
\end{align}

\end{proof}

\noindent Combining the last two lemmas we obtain the following corollary, which is used in the proof of the main result:

\begin{cor} \label{stepfct_approximation_fund_sol}
For every $\varepsilon > 0$ there exists $T > 0$ and $\delta > 0$, such that for every $U(t,.)$ with $t \geq T$ there exists $\chi \in \mathcal{T}_\delta(\R)$ with

\begin{align}
||U(t,.) - \chi||_{L^1(\R)} \leq \varepsilon.
\end{align}

\noindent Furthermore, $T$ and $\delta$ only depend on $\varepsilon$ and the bounds on $a$.
\end{cor}

\begin{proof}
Approximate H\"older continuity was already established, while tightness in $L^1$ follows from the estimate (\ref{Aronson_estimate}) of Theorem \ref{discrete_H_continuity_thm}. The dependence of the constants is easily checked revisiting the proofs of the previous two lemmas.
\end{proof}

\section*{Acknowledgement}

The author would like to thank Barbara Niethammer and Juan J. L. Vel\'azquez for inspiration, helpful discussions and proofreading. This work was supported by the German Research Foundation through the CRC 1060 \textit{The Mathematics of Emergent Effects}.

\bibliography{Nonlinear_Discrete_Coarsening_GrowthSolutions.bbl}
\bibliographystyle{alpha}

\end{document}